\documentclass[journal]{IEEEtran}

\usepackage{amsmath}
\usepackage{amssymb}
\usepackage{epsfig}
\usepackage{float}
\usepackage{multirow}
\usepackage{graphicx}
\usepackage{setspace}
\usepackage{subfigure}
\usepackage{cite}
\usepackage[ruled]{algorithm}
\usepackage{algpseudocode}
\usepackage{breqn}

\alglanguage{pseudocode}

\newtheorem{theorem}{Theorem}
\newtheorem{proposition}{Proposition}

\allowdisplaybreaks

\begin{document}
\title{Relay Selection for OFDM Wireless Systems under Asymmetric Information: A Contract-Theory Based Approach}
\author{Ziaul~Hasan,~\IEEEmembership{Member,~IEEE,}
        Vijay~K.~Bhargava,~\IEEEmembership{Fellow,~IEEE,}

\thanks{This research was supported by the Natural Sciences and Engineering Research Council of Canada (NSERC) under their strategic project award program and in part by Alexander Graham Bell Canadian Graduate Scholarship. This work was presented in part at IEEE WCNC 2012, Paris, France. Z. Hasan, and V. K.~Bhargava are with the Department of Electrical and Computer Engineering, University of British Columbia, Vancouver, Canada. (\{ziaulh, vijayb\}@ece.ubc.ca). Z. Hasan is thankful to Prof. Abbas Jamalipour, University of Sydney, NSW, Australia for inviting him to spend time in his lab and to Prof. Md. Jahangir Hossain for his valuable feedback and comments to improve the manuscript.}}

\maketitle

\begin{abstract}
User cooperation although improves performance of wireless systems, it requires incentives for the potential cooperating nodes to spend their energy acting as relays. Moreover, these potential relays are better informed than the source about their transmission costs, which depend on the exact channel conditions on their relay-destination links. This results in asymmetry of available information between the source and the relays. In this paper, we use contract theory to tackle the problem of relay selection under asymmetric information in OFDM-based cooperative wireless system that employs decode-and-forward (DF) relaying. We first design incentive compatible offers/contracts, consisting of a menu of payments and desired signal-to-noise-ratios (SNR)s at the destination. The source then broadcasts this menu to nearby mobile nodes. The nearby mobile nodes which are willing to relay, notify back the source with the contracts they agree to accept in each subcarrier. We show that when the source is under a budget constraint, the problem of relay selection in each subcarrier with the goal of maximizing capacity is a nonlinear non-separable knapsack problem. We propose a heuristic relay selection scheme to solve this problem. We compare the performance of our overall mechanism and the heuristic solution with a simple relay selection scheme. Selected numerical results show that our solution performs better and is close to optimal. The benefits of the overall mechanism introduced in this thesis is that it is simple to implement, needs limited interaction with potential relays and hence it requires minimal signalling overhead.
\end{abstract}


\section{Introduction}
Relay-assisted cooperation in wireless networks plays a key role in improving the overall efficiency of wireless networks by improving the system throughput, energy efficiency, spectrum usage, coverage, channel reliability and network cost reduction via spatial multiplexing and achieving diversity gains. Relay-based cellular network architectures have also been considered  for next generation wireless systems such as 3GPP Long Term Evolution (LTE) and IEEE 802.16j mobile WiMAX \cite{3gpp, 802}. 

Cooperation via relays can potentially assist a source node by forwarding its data to the destination either by amplify-and-forward (AF) or by decode-and-forward (DF) relay protocols. This cooperation can be achieved either by installing fixed relays within the network coverage area or by making the other mobile nodes act as relays. The latter scenario, also known as user cooperation, is gaining attention, because of the minimal changes required in existing infrastructure and because it has been shown to not only increase the data rates but also to make the achievable rates less sensitive to channel variations \cite{zhuang, sendonaris}. While user cooperation eliminates the cost of installing additional relay nodes, it increases the complexity of the overall system for several reasons. First, various dynamic resource allocation algorithms require near complete channel state information (CSI) from potential users assisting as relays. In the absence of this information, it is a challenge to design algorithms that dynamically select mobile users as potential relays \cite{sadek}. Most relay selection algorithms for cooperative networks assume complete CSI \cite{krikidis, beres}. However, this information is private to mobile users and they may not be willing to share this information. This results in an {\em asymmetry} of available information between the source mobile user and the potential relays. Secondly, user cooperation poses a logistic challenge because the increased rate of one user comes at the expense of consumption of the limited resources of the relaying user (e.g. battery, power, bandwidth etc.). The potential relays are usually selfish nodes that could belong to different network entities/operators and hence may not be willing to cooperate without any additional incentives. 

While relay selection with partial CSI have been explored by several authors e.g. \cite{ahmed, prs}, no incentive-based mechanisms has been considered in these and other related works. To tackle this problem, game-theoretic models have generally been suggested for cooperative systems that are either reputation-based, resource exchange-based, or pricing based \cite[references therein]{dejun}. However, there are still many challenges in applying game theoretic solutions to cooperative systems including investigating the existence, uniqueness, computation and efficiency of the Nash Equilibrium, as well as addressing signalling overheads \cite{dejun}. Moreover, to the best of our knowledge, the problem of relay selection under asymmetric information together with incentive-based mechanisms for OFDM-based cooperative systems, had not been studied. Therefore, the objective of this paper is to address this problem with simple pricing-based incentive mechanisms with minimal signalling overheads. 

Orthogonal Frequency Division Multiplexing (OFDM) has been adopted by many modern communication systems as standard multi-carrier modulation technology due to its ability to handle severe channel conditions without complex equalization filters. In a dual-hop OFDM-based cooperative communication system, choice of a relaying protocol, i.e., the interface between source-relay land the relay-destination links, is an important factor in defining the performance and complexity of such systems \cite{riihonen}. The authors in \cite{riihonen} present and analyze a comprehensive set of relaying protocols namely; time-domain versus frequency-domain processing, resource block-wise versus symbol-wise processing, AF versus DF, reordering/pairing of resource blocks, buffering over fading states, and optimization of time sharing. In this paper, we consider a OFDM-based dual-hop cooperative communication system that uses DF relaying due to its obvious performance advantages over AF relaying. We assume no buffering and equal time sharing for the two hops because of the low complexity and simple analysis. Pairing of subcarriers has not been considered explicitly because under the assumptions that we discuss later, this information is not required for the proposed solution to the relay selection problem under asymmetric information.

We use a simple {\em principal-agent model} from microeconomics for source and relay, where source acts as the principal and relay is an agent \cite{herbert}. In such a model, the bargaining power is kept with the principal, and the agent can either accept or decline an offer proposed by the principal. Such a model reduces the interaction needed between a source and a relay, and only the users willing to relay can participate. These offers or {\em contracts} proposed to relays are in the form of a menu of renumeration or monetary transfers for a specified service in each subcarrier, which in this paper is taken as desired signal-to-noise ratio (SNR) at destination. Assuming that the source has limited information about the relay-agents (e.g. joint probability distribution of CSI), we use ``contract theory" to design these set of {\em contracts} which are designed for different {\em types} of users (we will elaborate on ``types" later on in this paper). Contract theory is a field of economics that studies how economic players or agents create mutually agreeable contracts or arrangements in presence of asymmetric or incomplete information \cite{salanie}. Contract-based solutions for spectrum sharing in wireless systems have recently been discussed in \cite{duan, kalathil, kalathil2, gao}. The challenge behind contract-based approach for wireless systems is due to the fact that the agents can lie to principal about their individual information in order to increase their utility and hence a contract in this situation tries to create incentives for agents to report their information truthfully.

Once the prospective relays confirm to the source which contracts (i.e. payments and SNRs) they are willing to accept for each subcarrier, the only problem that remains is to select appropriate relays in each subcarrier such that the overall capacity is maximized for the source under overall budget constraint. Without designing these contracts and having no information about exact channel conditions of the potential relays, it would be difficult for the source to optimally choose relays and offer them suitable renumeration. This is because relays can lie about their channel conditions and hence the source could possibly make inefficient payments with unsatisfactory performance. As we will see later in this paper, this relay selection now becomes a nonlinear non-separable convex knapsack problem which cannot be solved with reasonable computational complexity and we therefore suggest a heuristic method to solve it. We then compare the performance of our overall mechanism and heuristic relay selection solution with a simple relay selection scheme and show that not only the proposed solution performs better but also it is near optimal under most general settings. The proposed scheme is very simple to implement in practical systems, requires almost no information about relays at source, limits the computational overheads only at the source and requires very limited interaction with potential relays.

A list of terms and definitions used in the paper is given in Table \ref{tb:terms}. The organization of the paper is as follows. We will first present our system model, problem description and solution approach in Section \ref{model}. The utility models for contract design will be discussed in Section \ref{util}. In Section \ref{fbcase}, we will discuss a contract formulation under a complete or perfect information. In Section \ref{sbcase}, we will show how to obtain an optimal contract design under asymmetric information. Relay section under a budget constraint will be discussed in Section \ref{relayselection}. Numerical results will be presented in Section \ref{simulation} and conclusions will be drawn in Section \ref{concl}.

\begin{table*}[ht]
\begin{center}
\caption{Terminology Used}
\label {tb:terms}
\begin{tabular}{p{3.4cm}|p{13cm}}
\hline
{\em principal-agent model} & microeconomic model where payoff to the principal (source) depends on an action taken by the relay-agent and the bargaining power is kept with the principal.\\
\hline
{\em asymmetric information} & one party (relay-agent) has more or better information than the other (source).\\
\hline
{\em contract} & a tuple consisting of a targeted SNR that a relay of a certain type on a particular subcarrier can provide on that subcarrier, and a corresponding payment that the source promises to make to that relay.\\
\hline
{\em private information} & relay agent's instantaneous channel gains between relay-destination link on all subcarriers.\\
\hline
{\em type} & relay agentÕs private information on a particular subcarrier, i.e., channel gain on a relay-destination link on that subcarrier.\\
\hline
{\em transfer} & payment made to a relay node by a source in lieu of a targeted SNR at destination.\\
\hline
{\em reservation utility} & minimum utility the relay-agent will get by not accepting a contract.\\
\hline
{\em incentive compatible} & relay-agent chooses the contract designed for his type only.\\
\hline
{\em individually rational} & contract designed for each type gives the relay-agent at least as much utility as it would get by not accepting the offer.\\
\hline
{\em revelation principle} & Every equilibrium outcome with a mechanism is realized by a payoff-equivalent revelation mechanism that has an equilibrium where the relay-agents truthfully report their types.\\
\hline
{\em single crossing property} & indifference curves for two different types of relay agents cannot intersect more than once.\\
\hline
{\em information rent} & positive surplus that the relay receives by accepting a contract.\\
\hline
\end{tabular}
\end{center}
\end{table*} 

\section{System Model}
\label{model}
\subsection{Problem Description}
We consider a typical cooperative network scenario in which a particular mobile node acting as a source, wants to transmit a block of data to a destination node with the help of some nearby mobile nodes that can act as relays. Fig. \ref{fig:systemmodel} shows an example of such a cooperative network system.  We assume that the source uses an OFDM-based multi-carrier system for the transmission technology with a total number of $N$ subcarriers. We posit that there are $M$ such mobile stations that could be the possible relay candidates. Based on its channel conditions, each relay node incurs a certain cost to provide a pre-specified SNR to the source on a particular subcarrier at the destination. Since there is no obligation for these mobile stations to forward its data towards the destination, the source mobile node must provide some incentives such as some monetary payment or credits, to these possible relay candidates. However, in such a system, it is practical to assume that the source is not only unaware of the number of possible relay candidates $M$ but also the exact channel conditions on all the relay-destination links on each subcarrier. In absence of this information, the source does not know which relay nodes to choose and how much it should pay to each relaying node because relaying nodes are regular mobile users and are therefore selfish, and can potentially lie about their actual cost of transmission. The relay node's {\em private information}, i.e., instantaneous channel gains between relay-destination link on all subcarriers can be expressed as a vector $\boldsymbol{\Psi}_m=\{\theta_{1m}, \theta_{2m}, \cdots \theta_{Nm}\}$, where $\theta_{im}$ denotes the channel gain for the $m$th relay on the $i$th subcarrier for the corresponding relay-destination link. A relay node's channel gain on a relay-destination link on a particular subcarrier will henceforth be called its {\em type}\footnote{To avoid confusion, we would like to clarify that {\em type} of a relay is not to be confused with relay types in LTE systems which classifies relay nodes as type 1, 1a, 1b or 2\cite{yangyang}. In this paper, usage of {\em type} is synonymous to the equivalent standard definitions in the theory of economics of asymmetric information \cite{salanie}.}
on that subcarrier and we will subsequently use the symbol $\theta$ without subscripts to indicate types in general. We assume these channel gains to be slow-varying, which means that they would remain constant for both transmission and relaying time slots.

\begin{figure}
\centering
\includegraphics[width=3.5in]{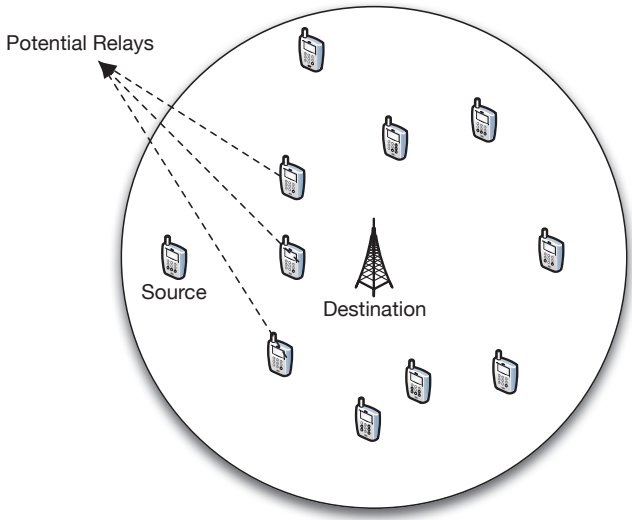}
\caption{A typical cooperative wireless network system with user cooperation}
\label{fig:systemmodel}
\end{figure}

Although the source is unaware of relay node's exact type or channel condition, we assume that it has information about the joint distribution of the types and the set $\boldsymbol{\Theta}\subset\mathbb{R}^N_{\geq0}$ from which these type vectors are drawn from. This is a reasonable assumption because source can learn about this distribution through the knowledge of fading environment parameters between relays and destination and these parameters could be provided to the source with limited feedback from the destination. However, as we will later observe that knowledge of this distribution only affects the optimality of the designed solution hence the source can begin transmission with just a priori belief about this distribution. 

Moreover, we assume that the source has a maximum budget $\mathcal{T}$ in one time frame for the total payments or {\em transfers} that it can make to the relay nodes over all subcarriers. We further assume that the relays utilize space-time coded cooperative diversity for multi-relay transmission on the same subcarrier. Based on these assumptions, the problem can be described as follows: the source has to effectively choose a set of relays for every subcarrier in order to maximize its overall throughput in a given time frame, make optimal transfers to these relays without the knowledge of relays' private information, while making sure total transfers do not exceed the overall budget constraint $\mathcal{T}$. 

\subsection{Two Part Contract-based Solution}
\label{sub:twopart}
This overall problem of selecting relays on each subcarrier while providing transfers to different relays is quite difficult to solve because of the overall budget constraint, multi-dimensional information types and information asymmetry (difference in available information between the source and the potential relays). However with the help of contract theory, we will attempt to solve it by breaking the problem down into two parts: 
\begin{enumerate}
\item {\em Contract Design:} A {\em contract}, is defined as a tuple consisting of a targeted SNR at destination on a certain subcarrier that a relay can provide and a corresponding guaranteed transfer or monetary incentive that source promises to make. The source first designs a set of common contracts applicable to all subcarriers without using any specific budget constraints, and broadcasts them to all relay candidates using some established protocol. Since, the number of contracts broadcasted are independent of number of subcarriers, the signalling overhead will be small. The relays listen to the contracts offered by source then respond back identifying the contracts they are willing to accept for each subcarrier. This part also has less signalling overhead than regular communication because relays do not have to provide their actual channel feedback.

\item {\em Relay Selection:} Based on the accepted contracts by the relays on each subcarrier, source then chooses an optimum set of relays for each subcarrier considering its budget constraint, while maximizing its expected capacity. We assume that the source instructs the selected relays with space-time-codes for each subcarrier, and hence relays can transmit simultaneously on the same subcarrier. Each relay node hence acts as a ``virtual antenna" and sends out the signal in a Multiple Input Multiple Output (MIMO) setting. All the relay transmissions hence occur in the same subcarrier and superimpose at the destination so that the overall SNR in each subcarrier is the summation of SNRs provided by all selected relays \cite{laneman,molisch}.
\end{enumerate}
An important assumption for this proposed two-part contract based solution is that the overall budget constraint $\mathcal{T}$ is sufficiently larger than the average cost of transmission for relay in each subcarrier. This assumption is reasonable because source expects the transmission to happen at least on a few subcarriers, hence source must set a budget constraint several times larger than the average cost of transmission for relay in each subcarrier. This assumption ensures that the monetary incentive or transfer for any contract pair would be sufficiently smaller that $\mathcal{T}$.  Fig. \ref{fig:relayselection} illustrates the complete representation of this mechanism in a three step process for a sample system.
\begin{figure*}[ht]
\centering

\subfigure[Step 1: Source S broadcasts a set of contracts to all nearby relays M1, M2, M3, M4.]{
\includegraphics[width=2.0in]{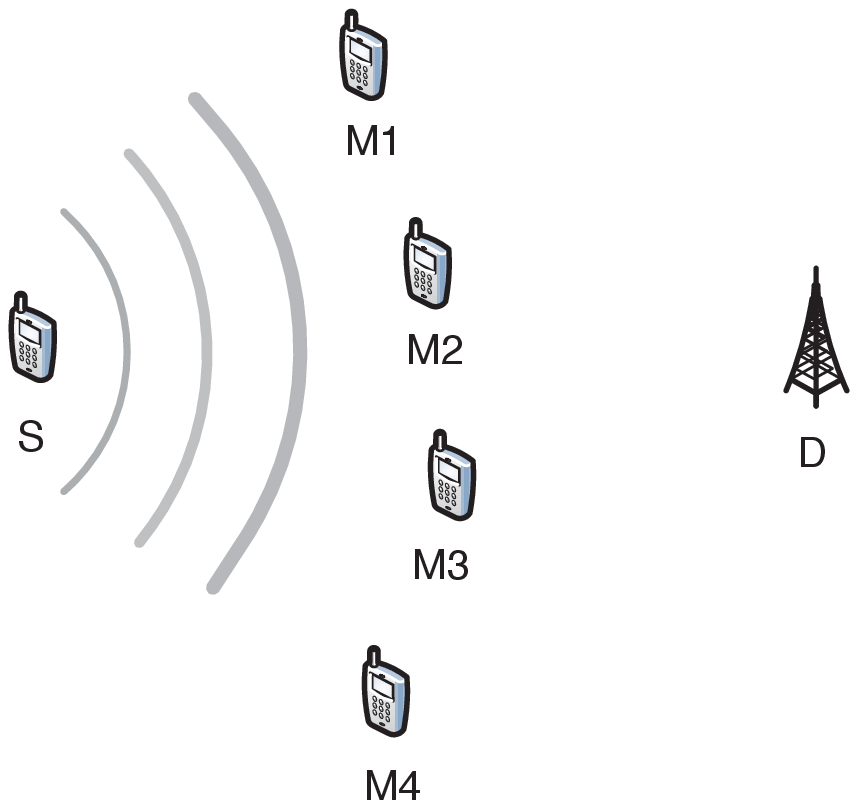}
\label{fig:sm1}
}
\subfigure[Step 2: Relays M1, M2 and M3 accept certain contracts on certain subcarriers and respond to source indicating which contracts they are willing to accept.]{
\includegraphics[width=2.0in]{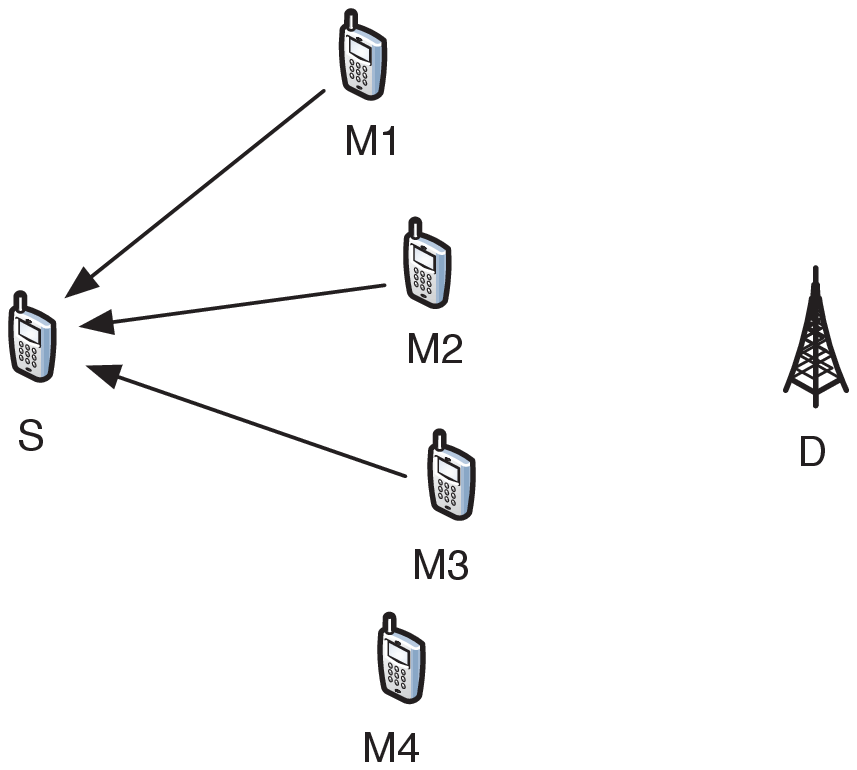}
\label{fig:sm1}
}
\subfigure[Step 3: Source selects relays M1 and M2, makes transfers, sends instructions and messages, and relays help transmit to destination D.]{
\includegraphics[width=2.0in]{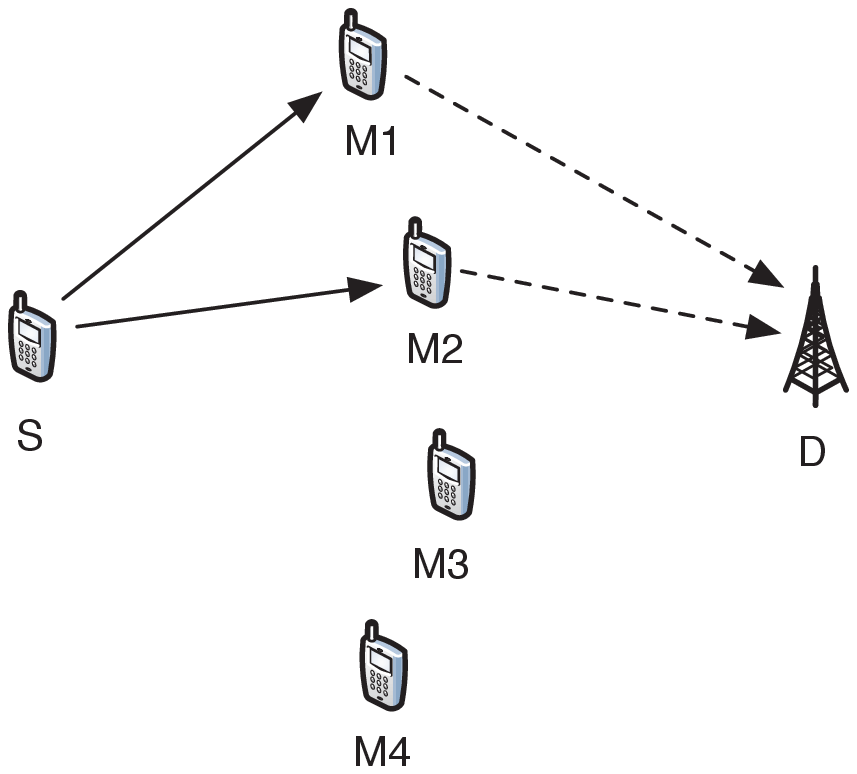}
\label{fig:sm3}
}
\caption[Optional caption for list of figures]{ A contract based cooperative communication and relay selection mechanism}
\label{fig:relayselection}
\end{figure*}

\section{Utility Models for Contract Design}
\label{util}
As we explained in section \ref{sub:twopart}, a {\em contract} is defined as a tuple consisting of a targeted SNR that a relay of a certain type on a particular subcarrier can provide on that subcarrier, and a corresponding payment that the source promises to make to that relay. For the purpose of contract design, we will focus our analysis for a general relay type $\theta$ (channel gain $\theta$ on the relay-destination link) on any subcarrier. To begin with, source first needs to design a set of common contract pairs $(\gamma(\theta), t(\theta))$ for a given range of types $\theta$ that are {\em incentive compatible} (IC) and {\em individually rational} (IR) on all subcarriers. Here, $\gamma(\theta)$ is the SNR that the relay of type $\theta$ can provide at the destination and $t(\theta)$ is the transfer (e.g. monetary incentive) that the source node makes to the relay for a particular subcarrier. {\em Incentive compatible} means that the relay-agent chooses the contract designed for his type only. {\em Individually rational} means that the contract designed for each type gives the relay-agent at least as much utility as it would get by not accepting the offer or in other words, by not relaying. This minimum utility is also known as {\em reservation utility} and we will take it as 0 in the rest of the paper. Using the {\em revelation principle}, we can focus our analysis on contract designs where agents declare their types truthfully or in other words, we can directly consider types while designing contracts \cite{salanie}. We will now make a practical assumption that the relay nodes do not have a prior knowledge of source's budget $\mathcal{T}$ and the number of other relay candidates. In other words, the relay nodes have no way of knowing that whether the contract they will accept will be executed or not by the source. Assuming DF relaying with repetition coding, the SNR of the source-relay-destination link on a certain subcarrier will be given by $\min\{\gamma_{SR},\gamma_{RD}+\gamma_{SD}\}$, where $\gamma_{SR}$ is source-relay SNR and $\gamma_{RD}$ is relay-destination SNR \cite{molisch}. Ignoring the direct path, the source-relay link can be safely assumed to be stronger because source will only consider nearby mobile nodes who can decode the source information, as relay candidates and hence $\gamma_{RD}$ will be the bottleneck \cite{duan}. Moreover, due to the assumption that source-relay link is always stronger, subcarrier pairing has also been ignored in the current problem formulation. Therefore from now on, we will approximate $\gamma_{RD}$ as the net SNR for the overall link and will drop the subscript and call it just $\gamma$ henceforth. 

We now define a quasi-separable utility for the source with a contract pair $(\gamma(\theta), t(\theta))$ on each subcarrier as follows:
\begin{equation}
U(\gamma(\theta))-t(\theta)
\end{equation}
where $U(\cdot)$ is a concave function that gives the utility that the source gets with an SNR $\gamma$ between a relay and destination on some subcarrier, when it is using only this particular relay, and this utility can be given by Shannon capacity formula:
\begin{equation}
U(\gamma)={1\over 2}\log_2\left(1+\min\{\gamma_{SR}, \gamma+\gamma_{SD}\}\right)\approx{1\over 2}\log_2\left(1+\gamma\right).
\end{equation}
Here, the half factor is used to account for half duplexity of the relaying protocol and approximation is based on the argument in the previous paragraph. The overall utility of a relay candidate of type $\theta$ that announces its type truthfully can also be described by a quasi-separable function given by the difference between transfers and cost of transmission:
\begin{equation}
t(\theta)-\mathcal{C}(\gamma(\theta),\theta)
\end{equation}
where $\mathcal{C}(\gamma,\theta)$ is the cost for relay of type $\theta$ to provide SNR $\gamma$ at the destination on some subcarrier. This cost could be the summation of the cost of per unit power used for relaying in addition to fixed decoding costs. Ignoring the decoding costs, this cost can simply be given by:
\begin{equation}
\mathcal{C}(\gamma, \theta)={c\gamma\over\theta}
\end{equation}
where $c$ is a positive number denoting cost per unit power and $\gamma\over\theta$ is the transmitted power by the relay. Conveniently, since ${\partial^2\mathcal{C}(\gamma, \theta)\over\partial\gamma\partial\theta}<0$, the cost function satisfies a form of Spence-Mirrlees Condition or {\em single crossing property} \cite{salanie}. What this means is that indifference curves $(\gamma, t)$ (plot of contracts for which a relay gets constant utility) for two different types of relay agents cannot intersect more than once. Moreover, the economic significance of this condition is that the relay agents of higher types are willing to provide better SNR $\gamma$ for a smaller increase in transfer $t$. In the following sections, we will see how to solve the problem of contract design for two specific cases: complete information scenario and incomplete information scenario.

\section{Contract Formulation under Complete and Asymmetric Information}
\subsection{Contract Formulation under complete Information: First-best scenario}
\label{fbcase}
This is the scenario when the source has the precise information of the relay types vector $\boldsymbol{\Psi}_m$, i.e., CSI between relay-destination link in each subcarrier. This is also known as the first-best scenario and will be used for benchmarking, because this is the ideal case for source, as it is in the best position to make the maximum use of the available potential relays. Therefore in this scenario, source only has to ensure that each relay agent is ready to accept the contract that he is about to offer, or in other words he only has to satisfy relay-agent's individual rationality condition in each subcarrier. The source's objective problem for a relay-agent type of $\theta$ in a certain subcarrier can now be written as:
\begin{eqnarray}
\label{optfb}
\max_{\gamma(\theta),t(\theta)} U(\gamma(\theta))-t(\theta)\\
\label{sbfb}
\text{subject to}\hspace{0.3 in} t(\theta)-{c\gamma(\theta)\over\theta}\ge0.
\end{eqnarray}
The source hence gives the relay-agent zero utility in order to maximize its own utility, i.e., the source extracts all the surplus from the relay-agent. Therefore, setting (\ref{sbfb}) to equality, substituting in (\ref{optfb}) and then differentiating the objective w.r.t $\gamma$ and finally equating it to zero, the optimal first-best contract $(\gamma(\theta),t(\theta))$ for a relay of type $\theta$ is given by $({\theta\over 2c\ln2}-1, {1\over2\ln2}-{c\over\theta})$. 

\subsection{Contract Formulation under Asymmetric Information: Second-best scenario}
\label{sbcase}
\subsubsection{Theoretical analysis with continuous relay-agent types}
In this section, we will analyze the structure of the solution using some standard theoretical analysis. In this case, we assume that the types of relays for all subcarriers are continuous and belong to a set $\Theta=[\underline{\theta},\bar\theta]$ and has a joint probability distribution $f(\theta_1,\theta_2, \cdots\theta_N)$ (with $F(\theta_1,\theta_2,\cdots\theta_N)$ as cumulative density function), which is known to the source node.
Let $\mathcal{P}(\hat\theta,\theta)$ be the profit or utility achieved by relay agent of type $\theta$ on a certain subcarrier who announces his type as $\hat\theta$. The profit is given by the following function:
\begin{equation}
\mathcal{P}(\hat\theta,\theta)=t(\hat\theta)-\mathcal{C}(\gamma(\hat\theta),\theta).
\end{equation}
The contract $(\gamma(\theta),t(\theta))$ satisfies the incentive constraints (IC) if and only if being truthful gives a relay node at least as much utility as it gets by lying, i.e.,
\begin{equation}
\mathcal{P}(\theta,\theta)\ge\mathcal{P}(\hat\theta,\theta), \;\;\;\;\forall (\theta,\hat\theta)\in \Theta^2 \hspace{0.3 in}\text{(IC)}.
\end{equation}
Hence, for the contract to be {\em incentive compatible}, the following first and second order conditions must hold:
\begin{eqnarray}
\forall \theta\in \Theta,
 \begin{cases}
{\partial \mathcal{P}(\hat\theta,\theta) \over \partial\hat\theta}\lvert_{\hat\theta=\theta}=0 \hspace{0.5 in}&\text{(IC$_1$)}\\
{\partial^2 \mathcal{P}(\hat\theta,\theta) \over \partial\hat\theta^2}\lvert_{\hat\theta=\theta}\le0 \hspace{0.5 in}&\text{(IC$_2$)}.\\
 \end{cases}
\end{eqnarray} 
 Substituting for $\mathcal{P}(\hat\theta,\theta)$,  these conditions can be simplified to
\begin{eqnarray}
\forall \theta\in \Theta,
  \begin{cases}
{dt(\theta)\over d\theta}={c\over\theta}{d\gamma(\theta)\over d\theta} \hspace{0.5 in}&\text{(IC$_1$)}\\
{d\gamma(\theta)\over d\theta}\ge0 \hspace{0.5 in}&\text{(IC$_2$)}.\\
 \end{cases}
 \end{eqnarray}
 This means that first both $\gamma(\theta)$ and $t(\theta)$ must be increasing in type $\theta$ (by IC$_2$) and secondly, IC$_1$ tells us how the increase in transfers w.r.t to the agent types are related to increase in deliverable SNR. Let $\rho(\theta)$ denotes the utility of the relay agent of type $\theta$ with the optimal truthful contract, i.e, a mechanism where relay chooses contract designed for his type only. Then, $\rho(\theta)$ can simply be given by $\mathcal{P}(\theta,\theta)$, i.e.,
\begin{equation}
\rho(\theta)=t(\theta)-{c\gamma(\theta)\over\theta}.
\label{rho}
\end{equation}
Using IC$_1$, we can compute that
\begin{equation}
{d\rho\over d\theta}={dt(\theta)\over d\theta}-{c\over\theta}{d\gamma(\theta)\over d\theta}+{c\gamma(\theta)\over\theta^2}={c\gamma(\theta)\over\theta^2}
\label{drho}
\end{equation}
which is positive and implies that $\rho(\theta)$ is an increasing function of $\theta$ and hence the higher types benefit with higher returns. Assuming relay agent's {\em reservation utility} to be 0, its {\em individual rationality} (IR) condition can therefore simply be given by:
\begin{equation}
\rho(\underline\theta)=0   \hspace{0.3 in}\text{(IR).}
\label{irc}
\end{equation}
This is because making transfers is costly to the source node and since higher type relay nodes have higher returns, source just has to give zero utility to the lowest type $\underline\theta$ to satisfy the IR condition. Using equations (\ref{rho}), (\ref{drho}), and (\ref{irc}), we can hence write
\begin{equation}
t(\theta)={c\gamma(\theta)\over\theta}+\int_{\underline\theta}^{\theta}{c\gamma(\tau)\over\tau^2}d\tau.
\label{tranc}
\end{equation}
Now the source's objective is to maximize the expected utility which is given as:

\begin{equation}
\begin{split}
\int_{\underline\theta}^{\bar\theta}\int_{\underline\theta}^{\bar\theta}\cdots\int_{\underline\theta}^{\bar\theta}\sum_{n=1}^{N}(U(\gamma(\theta_n))-t(\theta_n))f(\theta_1,\theta_2,\cdots,\theta_N)\\d\theta_1d\theta_2\cdots d\theta_N.
\label{sobjective}
\end{split}
\end{equation}

\begin{proposition}
\label{prop:contsobjective}
We can rewrite source's optimization problem as follows: 
\begin{dmath}
\label{ctobjective}
\max_{\gamma(\theta)}\sum_{n=1}^{N}\int_{\underline\theta}^{\bar\theta}\left(U(\gamma(\theta))-{c\gamma(\theta)\over\theta}-{c\gamma(\theta)\over\theta^2}{1-F_n(\theta)\over f_n(\theta)}\right)f_n(\theta)d\theta\nonumber\\
\text{subject to IC$_2$ or     } {d\gamma(\theta)\over d\theta}\ge0 \;\; \text{(i.e. $\gamma$ is increasing)}\nonumber\\
\text{and $\gamma\ge0$,  (SNR must be positive)}
\end{dmath}
where $f_n(\theta)$ is the marginal probability distribution and $F_n(\theta)$ is corresponding cumulative distribution of types in the $n$th subcarrier. 
\end{proposition}
\begin{IEEEproof}
Proof is provided in Appendix.
\end{IEEEproof}
The optimization problem in (\ref{ctobjective}) can be interpreted as follows: the source has to maximize the expression in the brackets of the objective function subject to the constraint that $\gamma (\theta)$ is positive and increasing in $\theta$, where the first two terms of the objective are same as in the source's optimization problem under complete information scenario and the last term measures the impact of incentive problem. In order to solve the optimization problem in (\ref{ctobjective}), first we can just try to do pointwise maximization of the objective function at each $\theta$. However, if pointwise maximization at each $\theta$ in (\ref{ctobjective}) does not give us an increasing $\gamma(\theta)$ function, then we can use resort to optimal control theory. 

Using this continuous case as a reference, we will now focus our analysis to the case where the types are considered to be discrete values rather than taken from a continuous set. This case is of more practical interest because the contracts can be easily broadcasted as finite number of values. In the next subsection, we will design incentive compatible contracts by approximating the continuous distribution by a discrete distribution with finite points. With the discussion above, we have a good idea about the structure of the optimal contract and we will notice some parallels when we discuss the discrete agent types case in the following subsection. 

\subsubsection{Solution with discrete relay-agent types}
We will now see how to solve a more practical problem of designing contracts for relay agents with discrete types. This problem is more practical because the number of contracts are finite and can be transmitted to relay-agents in real-time. We quantize the set of types $\Theta=[\underline{\theta},\bar\theta]$ with a quantization factor $K$ such that the collection of types are reduced to a discrete set of $K$ types, i.e., $\Theta=\{\delta_1,\delta_2\cdots\delta_K\}$. Without loss of generality, we can assume that $\delta_1<\delta_2<\cdots<\delta_K$. We consider the quantization process to be uniform with equidistant values, i.e., $\delta_k=\underline{\theta}+\frac{k-1}{K}(\bar\theta-\underline{\theta})$, and if $\Theta$ is unbounded above, then $\bar\theta$ can be replaced by the upper limit of a desired confidence level. We chose quantization to be uniform mainly because of its ease of implementation and a closer representation of continuous distribution, however, in general a non-uniform quantization process can also be chosen depending upon how sensitive the cost function is to the variation in types. Using forward difference method, the probability that a relay-agent could be of type $\delta_k$ in $n$th subcarrier is given by $\pi_{kn}=P(\delta_k\le\theta_n<\delta_{k+1})=F_n(\delta_{k+1})-F_n(\delta_{k})$ ($\delta_{K+1}$ can be replaced by $\bar\theta$) with $\sum_{k=1}^{K}\pi_{kn}=1$. We assume that source is aware of this distribution on all subcarriers.

The objective of the source is to maximize its expected utility by designing an incentive compatible and individually rational optimal contract $(\gamma(\delta_k),t(\delta_k))$ (for simplicity, we will now refer it as $(\gamma_k,t_k)$) for each $\delta_k\in\Theta$, i.e.,
\begin{equation}
\begin{split}
\max_{\gamma(\theta),t(\theta)}\sum_{n=1}^{N}\mathbf{E}_n[U(\gamma(\theta))-t(\theta)]=\\\max_{\gamma_k,t_k \forall k}\sum_{n=1}^{N}\sum_{k=1}^{K}\pi_{kn}(U(\gamma_k)-t_k).
\label{objectorig}
\end{split}
\end{equation}
The {\em individual rationality} condition for this discrete scenario can now be given by:
\begin{equation}
t_k-{c\gamma_k\over\delta_k}\ge0,\;\;\;\forall\delta_k\in\Theta
\label{discreteIR}
\end{equation}
and the {\em incentive compatibility} condition is given by:
\begin{equation}
t_k-{c\gamma_k\over\delta_k} \ge t_j-{c\gamma_j\over\delta_k},\;\;\;\forall \delta_k,\delta_j\in\Theta.
\label{discreteIC}
\end{equation}

\begin{theorem}
\label{th1}
For the optimal solution, the individual rationality condition for the lowest type is binding, i.e., $t_1-{c\gamma_1\over\delta_1}=0$ and the others can be ignored.
\end{theorem}
\begin{IEEEproof} 
For any $\delta_k\in\Theta$ using the IC condition from equation (\ref{discreteIC}) we can write that
\begin{equation}
t_k-{c\gamma_k\over\delta_k} \ge t_1-{c\gamma_1\over\delta_k} \ge t_1-{c\gamma_1\over\delta_1}
\end{equation}
since $\delta_k>\delta_1\ge0$ and $c\gamma\ge0$. Therefore, if IR for $\delta_1$ is inactive, so will be IR for $\delta_k$. Hence, all the other IRs except for $\delta_1$, can be ignored. Now, if IR for $\delta_1$ is not binding, then all transfers $t_k$'s can be reduced by the same amount, having no effect on IC and hence increasing source's utility.
\end{IEEEproof}

\begin{theorem}
\label{th2}
For the optimal solution, $0\le\gamma_1\le \gamma_2\le\cdots\le \gamma_K$, and all the downward adjacent ICs are binding and others can be ignored, i.e.,
\begin{equation}
t_k-{c\gamma_k\over\delta_k}=t_{k-1}-{c\gamma_{k-1}\over\delta_k}, \;\;\;\forall k\ge2.
\label{adjacent}
\end{equation}
\end{theorem}
\begin{IEEEproof}
Proof is provided in Appendix.
\end{IEEEproof}

Now using theorems \ref{th1} and \ref{th2}, the optimization problem in (\ref{objectorig}), (\ref{discreteIR}), and (\ref{discreteIC}) can be reduced to:
\begin{eqnarray}
\max_{\gamma_k\forall k}\sum_{n=1}^{N}\sum_{k=1}^{K}\pi_{kn}(U(\gamma_k)-t_k)\nonumber\\
\text{s.t. }t_1={c\gamma_1\over\delta_1}, 
t_k={c\gamma_1\over\delta_1}+\sum_{i=2}^{k}{c(\gamma_i-\gamma_{i-1})\over\delta_i}\nonumber\\
\label{optimiz}
\text{and }0\le\gamma_1\le \gamma_2\le\cdots\le \gamma_K. 
\end{eqnarray}

\begin{proposition}
\label{prop:adverse1}
 The optimization problem in (\ref{optimiz}) can be rewritten as follows:
\begin{align}
\max_{\gamma_k\forall k}\sum_{n=1}^{N}\sum_{k=1}^{K}\pi_{kn} g_n(\gamma_k)
\;\;\text{s.t.}\;\;0\le\gamma_1\le \gamma_2\le\cdots\le \gamma_K,
\end{align}
where
\begin{dmath}
g_n(\gamma_k)=
 \begin{cases}
 U(\gamma_k)-{c\gamma_k\over\delta_k}-c\gamma_k\left(\frac{1}{\delta_k}-\frac{1}{\delta_{k+1}}\right)\left({1-\sum_{i=1}^{k}\pi_{in}\over\pi_{kn}}\right),\;\forall{k<K}\\
 U(\gamma_k)-{c\gamma_k\over\delta_K} \text{    if }k=K.
 \end{cases}\nonumber\label{reducedoptim}\\
\end{dmath}
\end{proposition}
\begin{proof}
Proof is provided in Appendix.
\end{proof}
Now because of the concavity assumption on $U(\cdot)$, this problem can easily be transformed into a simple convex optimization problem and can be solved using standard methods. In fact, it can be shown that the inequality constraint of increasing $\gamma$'s can also be neglected, and hence point-wise maximization of each $\sum_{n=1}^{N} g_n(\gamma_k)$ is sufficient here. Once we calculate contract pairs $(\gamma_k,t_k)$ $(k\le K)$, it can be easily verified that a relay that has a type $\theta$ in a certain subcarrier will automatically select contract $k$, if $\delta_k\le\theta<\delta_{k+1}$, because this contract will maximize its utility.

The objective of the contract design so far was to design incentive compatible and individually rational offers that the interested relay-agents can accept without revealing their types or channel information directly to the source. This helps the source to segregate the relays in terms of their abilities to deliver certain SNRs at the destination and the price at which they are willing to do so. Hence, the source can select relays based on his budget constraint and we discuss this relay selection problem and its solution in next section in detail.

\section{Relay Selection under a Budget constraint}
\label{relayselection}
In this section, we will discuss the relay selection procedure for the source under the budget constraint with either complete or incomplete information. While in perfect information scenario, the source is aware of the contracts acceptable by each relay in each subcarrier, in case of imperfect information, source broadcasts the contract menu to all relays and each relay responds with a desired contract pair for each subcarrier as discussed in section \ref{sbcase} (if a certain relay is unwilling to relay in a certain subcarrier, we assume it accepts null contract (0, 0)). In either case, the source knows a contract pair $(\gamma_{mn}, t_{mn})$ that is acceptable by relay $m$ in the $n$th subcarrier. Under a budget constraint $\mathcal{T}$, the objective of the source now becomes to maximize its total capacity. Let $\mathcal{M}=\{1,2,...,M\}$ denote the set of all the relay agents who are willing to relay while providing a certain SNR at destination in each subcarrier for a certain price that is determined by the contract for that relay agent. The objective is to obtain a subset vector $\mathbf{S}=\{\mathcal{S}_{n},\forall n\le N\,|\mathcal{S}_{n}\subseteq\mathcal{M}\}$ (i.e. set of selected relays in each subcarrier),
\begin{eqnarray}
\max_{\mathbf{S}}\mathcal{C}(\mathbf{S})\nonumber\\
\text{s.t. }\sum_{n=1}^{N}\sum_{\forall m\in\mathcal{S}_n}t_{mn}\le\mathcal{T}\nonumber\\
\text{where,   } \mathcal{C}(\mathbf{S})=\sum_{n=1}^{N}\log_2\left(1+\sum_{\forall m\in\mathcal{S}_{n}}\gamma_{mn}\right).
\label{optimprob}
\end{eqnarray}
This problem is a nonlinear non-separable convex knapsack problem in its current form with $\mathcal{T}$ as knapsack size, $t_{mn}$ as weights, $\gamma_{mn}$ as values of items, and $\mathcal{C}(\mathbf{S})$ as the objective function. Because of the non-linearity and non-separability of the objective function, it is difficult to find the exact solution in the existing form of this problem \cite{brett}. The standard method to obtain the optimal solution is to use branch and bound algorithm, where at each step, a series of continuous subproblems are solved to obtain upper bounds by integer relaxation of the original problem. The branch and bound algorithm is discussed in detail in \cite{hillier}. The authors in \cite{sharkey,romeijn} discuss methods to obtain these upper bounds for a general class of nonlinear non-separable knapsack problems. However, branch and bound method can still have the worst case complexity of exhaustive search and for $M$ number of subcarriers and $N$ number of potential relays, the worst case complexity could be as high as $O(2^{MN})$, which makes the branch and bound method practically infeasible due to exponential complexity. In fact, we verified the complexity to be exponential in most cases for a very simple system with simulations. Due to lack of space and practical importance, we will not go into further detail of obtaining the exact solution. Instead, in the next subsection, we will propose a heuristic solution based on the structure of our original problem described in (\ref{optimprob}). 

\subsection{Heuristic Solution}
\label{heuristicsol}
Here, we will discuss a few heuristics to solve the original problem by breaking it down in smaller problems that are standard 0-1 knapsack problems. We notice that if we could divide the overall constraint $\mathcal{T}$ by allocating a budget constraint $\mathcal{T}_n$ for subcarrier $n$, such that $\sum_{n=1}^{N}\mathcal{T}_n=\mathcal{T}$, then the sub-problem for relay selection in subcarrier $n$ is just a standard knapsack problem, with $\mathcal{T}_n$ as knapsack size, $t_{mn}$ as weights, and $\gamma_{mn}$ as values of items. This problem can just be written as follows:
\begin{eqnarray}
\max_{\mathcal{S}_n} \sum_{\forall m\in\mathcal{S}_{n}}\gamma_{mn} \;\;\text{ s.t. } \sum_{\forall m\in\mathcal{S}_n}t_{mn}\le\mathcal{T}_n.
\end{eqnarray}
We also notice that the objective of the problem is to maximize the product $\Pi_{i=1}^{n}(1+\sum_{\forall m\in\mathcal{S}_{n}}\gamma_{mn})$, subject to the budget constraints. In order to maximize this product we need to maximize SNR $\sum_{\forall m\in\mathcal{S}_{n}}\gamma_{mn}$ for each subcarrier, while also making sure that none of the SNR's are too low, otherwise they will minimize the product. Based on this analogy, we will attempt to solve this problem using two forms of heuristics and we combine the results to obtain the final solution.

\begin{algorithm*}[ht-y]
\caption{Relay Selection Algorithm}\label{rsdp}
\begin{algorithmic}[1]
\For{$i=1$ to $4$}
\State Set $\mathcal{S}_n^{(i)}=\phi, \forall n\le N$
\If {$i\le3$}
Obtain weights $w_n^{(i)}$ by (\ref{ee}), (\ref{aesw}) or (\ref{nesw}) for $i=1,2,3$ respectively $\forall n\le N$.
\State Set $\mathcal{T}_n=w_n^{(i)}\mathcal{T}/\sum_{n=1}^{N}w_n^{(i)}, \forall n\le N$
\For{$n=1$ to $N$}
\State Set $\tau=\mathcal{T}_n$
\For{$t=0$ to $\mathcal{T}_n$}
 \State$\Gamma[0,t]=0$
\EndFor
\For{$m=1$ to $M$}
  \For{$t=0$ to $\mathcal{T}_n$}
    \If {$t_{mn}\le t$ and $\Gamma[m-1,t-t_{mn}]+\gamma_{mn} > \Gamma[m-1,t]$} 
      \State $\Gamma[m,t] =\Gamma[m-1,t-t_{mn}]+\gamma_{mn}$,
       $s[m,t]=1$
    \Else 
    \State $\Gamma[m,t] = \Gamma[m-1,t]$, $s[m,t]=0$
    \EndIf
  \EndFor
\EndFor
\For{$m=M$ to $1$}
  \If {$s[m,\tau] = 1$}
    $\mathcal{S}_n^{(i)}:=\mathcal{S}_n^{(i)} \cup \{m\}$,
    $\tau=\tau-t_{mn}$
  \EndIf
\EndFor
\EndFor
\ElsIf{$i=4$}
\While{$\mathcal{T} \ge 0$}
\For{$n=1$ to $N$} 
\State $m=\arg\max_{\forall m\notin \mathcal{S}_n^{(i)}}\{\gamma_{mn}/t_{mn}\}$, 
 $\mathcal{T}=\mathcal{T}-t_{mn}$
\If {$\mathcal{T}\ge 0$}
$\mathcal{S}_n^{(i)}:=\mathcal{S}_n^{(i)} \cup \{m\}$
\EndIf
\EndFor
\EndWhile
\EndIf
\State Set $\mathbf{S}^{(i)}=\{\mathcal{S}_{1}^{(i)},\mathcal{S}_{2}^{(i)}\cdots \mathcal{S}_{N}^{(i)}\}$
\EndFor
\State \Return  $\mathbf{S}$ as given by (\ref{optimalset})
\end{algorithmic}
\end{algorithm*}

The first heuristic is to decompose $\mathcal{T}$ into $\mathcal{T}_n$ by using certain weights profile $w_n$ for subcarrier $n$ and hence, $\mathcal{T}_n$ could be given by $\mathcal{T}_n=\frac{w_n\mathcal{T}}{\sum_{n=1}^{N}w_n}$. The second heuristic will be to sequentially select relays for each subcarrier as long as we are under the budget constraint. We suggest three weight profiles for the first heuristic as follows:
\begin{enumerate}
\item {\em Equal subcarrier weights (ESW)}: The simplest possible way is to choose equal weights for all subcarriers: 
\begin{equation}
w_n^{(1)}=1,\;\;\;\forall n.
\label{ee}
\end{equation}
\end{enumerate}
A simple way to measure two different contracts relative to each other is to compare how much is the SNR per unit price for each contract. This metric, i.e., SNR per unit price, can be defined as the efficiency and can be used to compare the subcarriers relative to each other either by averaging the efficiencies of contracts in each subcarrier or by calculating net efficiency in each subcarrier. The following weight profiles are based on this observation:
\begin{enumerate}
\setcounter{enumi}{1}
\item {\em Average subcarrier efficiency weights (ASW)}: If the efficiency of a contract accepted by relay $m$ in subcarrier $n$ is defined by $e_{mn}=\frac{\gamma_{mn}}{t_{mn}}$, then the weights $w_n$ of $n$th subcarrier are obtained by the following relation: 
\begin{equation}
w_n^{(2)}=\frac{\sum_{m=1}^{M}{e_{mn}}}{M},\;\;\;\forall n.
\label{aesw}
\end{equation}
\item {\em Net subcarrier efficiency weights (NSW)}: In this case, weights are obtained by calculating the net efficiency of a subcarrier, i.e., the ratio of maximum achievable SNR and the corresponding maximum transfer:
\begin{equation}
w_n^{(3)}=\frac{\sum_{m=1}^{M}{\gamma_{mn}}}{\sum_{m=1}^{M}{t_{mn}}},\;\;\;\forall n.
\label{nesw}
\end{equation}
\end{enumerate}
Using each of the three above mentioned weight profiles, we solve the original problem iteratively by solving $N$ standard 0-1 type knapsack problems. These knapsack problems can be solved using dynamic programming by scaling and rounding transfers $t_{mn}$ and budget $\mathcal{T}_n$ \cite{korte}. For each of the above mentioned weight profiles, we obtain solution subset vectors $\mathbf{S}^{(1)}$, $\mathbf{S}^{(2)}$  and $\mathbf{S}^{(3)}$ for each subcarrier. For the second heuristic solution, we perform an efficiency based relay selection as follows:
\begin{enumerate}
\setcounter{enumi}{3} 
\item {\em Sequential Subcarrier Contract Pair Allocation (SSCPA)}: In this heuristic, we allocate a relay to each subcarrier sequentially by choosing a relay that provides greatest efficiency $e_{mn}$ and has not been allocated in that subcarrier. We repeat this process until we run out of $\mathcal{T}$ and we obtain the corresponding solution subset vector and call it $\mathbf{S}^{(4)}$.
\end{enumerate}
The overall heuristic solution is hence chosen as the solution vector that gives highest capacity, i.e.,
\begin{equation}
\mathbf{S}=\arg\max_{\mathbf{S}^{(i)}|i=1,2,3,4}\{\mathcal{C}(\mathbf{S}^{(i)})\}.
\label{optimalset}
\end{equation}

Table \ref{tb:heuristics} summarizes the overall heuristic solution and Algorithm \ref{rsdp} describes the entire relay selection heuristic. The net complexity of the algorithm can be calculated as follows: The SSCPA heuristic is just sorting and selecting $MN$ contract pairs, hence the complexity can be $O(MN\log(MN))$. The three other heuristics based on weight profiles solve $N$ knapsack problems, each upper bounded by pseudo polynomial complexity of $O(M\mathcal{T})$ (assuming $\mathcal{T}$ is rounded). Hence the overall complexity will be $O(MN\mathcal{T})$, which is much easier to handle than branch and bound algorithm for small to medium values of budget constraint.

\begin{table*}[ht]
\begin{center}
\caption{List of Heuristics}
\label {tb:heuristics}
\begin{tabular}{c|p{6cm}|c}
\hline
Heuristic Method & Description & Solution Subset Vector\\
\hline
ESW & Divide $\mathcal{T}$ based on weight profiles $w_n^{(1)}=1,\;\forall n$ and solve $N$ 0-1 knapsack problems & $\mathbf{S}^{(1)}$\\
\hline
ASW & Divide $\mathcal{T}$ based on weight profiles $w_n^{(2)}=\frac{\sum_{m=1}^{M}{e_{mn}}}{M},\;\forall n$ and solve $N$ 0-1 knapsack problems & $\mathbf{S}^{(2)}$\\
\hline
NSW & Divide $\mathcal{T}$ based on weight profiles $w_n^{(3)}=\frac{\sum_{m=1}^{M}{\gamma_{mn}}}{\sum_{m=1}^{M}{t_{mn}}},\;\forall n$ and solve $N$ 0-1 knapsack problems & $\mathbf{S}^{(3)}$\\
\hline
SSCPA & Sequential Allocation per subcarrier as per maximum efficiency& $\mathbf{S}^{(4)}$\\
\hline
\hline
Overall Heuristic & Combination of ESW, ASW, NSW and SSCPA  & $\mathbf{S}=\arg\max_{\mathbf{S}^{(i)}|i=1,2,3,4}\{\mathcal{C}(\mathbf{S}^{(i)})\}.$\\
\hline
\end{tabular}
\end{center}
\end{table*} 

%\pagebreak
\section{Numerical results}
\label{simulation}
In this section, we will briefly present some of the numerical results. For simulations, we consider a system where the relay types $\theta$ are normalized, independent and uniformly distributed between 50 to 300 in each subcarrier. We consider uniform distribution for our simulations because uniform distribution is the maximum entropy probability distribution for any random variable contained in the distribution's support. It effectively means that source has no additional information about the types other than their support and is therefore a benchmark scenario. In order to generate a set of first and second best contracts, we quantize the range of types with a quantization factor $K$ to be 10. The number of subcarriers $N$ is chosen to be 16 and the parameter $c$ is taken to be 1. The simulation parameters and the corresponding first and second-best contracts are presented in Table \ref{tb:fbsb}. The first-best contracts are calculated when the source is completely aware of the relay-agent's discrete type (complete information). However, the first-best contracts are not incentive compatible (downward ICs do not hold) and it can be easily verified by plugging the parameters of Table \ref{tb:fbsb} in relay-agent's overall utility. On the other hand, the second-best contracts are incentive compatible by virtue of design. The IC conditions can be verified with the corresponding second-best contracts, i.e., any relay of type $\theta$ s.t. $\delta_k\le\theta<\delta_{k+1}$ will automatically pick the $k$th contract, because this contract will maximize its expected utility. Another noticeable difference between the first and second-best contract that can be seen from Table \ref{tb:fbsb} is that under incomplete information, except for the lowest type, the source pays more to get a certain SNR than what it would have gotten for lesser price under complete information. Moreover, in case of incomplete information, the source asks for sub-efficient SNRs from all the relay types except from the highest type. This is in order to provide incentive for higher types to not to choose a lower types' contract, and the related concept is called {\em information rent} \cite{bolton}. In simple words, information rent is the positive surplus that the relay receives and Table \ref{tb:fbsb} clearly indicates that higher type relay gets more positive surplus for the contract designed for its type.

\begin{table*}[ht]
\begin{center}
\caption{First and Second Best Contracts}
\label {tb:fbsb}
\begin{tabular}{c|c|c|c}
\hline
Relay types & First-best contract & Second-best contract & Information Rent\\
\& distribution &   $(\gamma^{(1)},t^{(1)})$ ($\gamma$ in dB) & $(\gamma^{(2)},t^{(2)})$ ($\gamma$ in dB) & $t^{(2)}-{c\gamma^{(2)}}/{\theta}$\\
\hline
$\delta_1=50, \pi_{1n}=0.1$& (15.4490,0.7013) & (9.0401,0.1603) & 0\\
$\delta_2=75, \pi_{2n}=0.1$& (17.2510,0.7080) &(12.3131,0.2806) &0.0534\\
$\delta_3=100,\pi_{3n}=0.1$& (18.5208,0.7113) &(14.6324,0.4008) &0.1102\\
$\delta_4=125,\pi_{4n}=0.1$& (19.5021,0.7133) &(16.4428,0.5210) &0.1683\\
$\delta_5=150,\pi_{5n}=0.1$& (20.3020,0.7147) &(17.9322,0.6412) &0.2271\\
$\delta_6=175,\pi_{6n}=0.1$& (20.9773,0.7156) &(19.1990,0.7615) &0.2863\\
$\delta_7=200,\pi_{7n}=0.1$& (21.5615,0.7163) &(20.3020,0.8817) &0.3457\\
$\delta_8=225,\pi_{8n}=0.1$& (22.0764,0.7169) &(21.2794,1.0019) &0.4052\\
$\delta_9=250,\pi_{9n}=0.1$& (22.5367,0.7173) &(22.1564,1.1221) &0.4649\\
$\delta_{10}=275,\pi_{10n}=0.1$& (22.9528,0.7177) &(22.9528,1.2424) &0.5246\\
\hline
\end{tabular}
\end{center}
\end{table*} 

Next, we will evaluate the performance of four heuristics, namely ESW, ASW, NSW, SSCPA, that we suggested in section \ref{heuristicsol} with respect to each other and a few benchmarks that we will describe here. We compare the performance of these heuristics with the obvious simple solution, i.e., select the contracts that offer the best SNRs amongst all contracts for all subcarriers while satisfying the budget constraint. Moreover, we will also compare the performance of the proposed heuristics with respect to the solution of original problem in (\ref{optimprob}) with relaxed integer constraints as a benchmark. We will call these solutions as {\bf ``Best SNR contracts"} and {\bf``Relaxed Solution"} respectively in the corresponding plots. In Figures \ref{h1}, \ref{h2}, and \ref{h3}, we plot average capacity per subcarrier vs number of relay agents for three values of $\mathcal{T}$, i.e., 8, 16 and 24 respectively. The number of subcarriers $N$ are fixed at 16 and quantization factor $K$ is chosen to be 10. As we could notice from these plots, the heuristic SSCPA always performs better when there are fewer relay agents or when the budget constraint is large. The three other heuristics based on weight profiles, namely ESW, ASW and NSW have a very similar behavior and perform better than SSCPA when number of relay agents are high and budget is not too big. The intuitive reasoning behind this observation is 
as follows: As the number of relay agents increase, there are more diversified contracts available per subcarrier to choose from. The ESW, ASW and NSW schemes by their inherent design try to maximize the product $\Pi_{i=1}^{n}(1+\sum_{\forall m\in\mathcal{S}_{n}}\gamma_{mn})$ by splitting the budget in each subcarrier and hence maximizing the sum SNRs in every subcarrier. Under low to medium budget conditions and with high number of relay agents, these algorithms outperform SSCPA because the latter scheme could run out of budget before it could select contracts in each subcarrier. Moreover under such conditions, SSCPA performs poorly as the number of agents increase because sequential allocation may result in first choosing contracts that may need higher transfers and hence source may run out of budget too quickly without balancing sum SNRs well. This behavior can be seen numerically in Figures \ref{h1} and \ref{h2}. Under high budget conditions, SSCPA scheme has more freedom to sequentially choose best and optimum contracts per subcarrier as long as the budget allows inherently improving the product $\Pi_{i=1}^{n}(1+\sum_{\forall m\in\mathcal{S}_{n}}\gamma_{mn})$ while automatically balancing the SNRs per subcarrier. Fig. \ref{h3} demonstrates this adequately. Moreover, ``Best SNR Contracts" solution not only has inferior performance compared to proposed heuristics in general, but the average capacity with this solution decreases as the number of relay agents increase. This is because the ``Best SNR Contracts" solution just selects the contracts that offer the best SNRs amongst all subcarriers without actually balancing the SNRs amongst all subcarriers reducing the overall capacity. In addition to this, the gap between the envelope of proposed heuristics (overall heuristic solution) and ``Relaxed Solution" decreases as $\mathcal{T}$ is increased. This gap reduces with budget because with higher budget more contracts can be chosen as whole per subcarrier hence reducing the difference in capacity obtained with the ``Relaxed Solution". Notice that the optimal solution lies in between this gap, hence, smaller this gap is, better is the performance. Additionally, we notice that for the overall proposed heuristic, capacity tends to converge to a stable value as number of relay agents are increased. The convergence happens because of the diversification of independent relay types.

%\begin{figure}
%\centering
%\includegraphics[width=3.7in]{images/h1}
%\caption{Comparison of different heuristics with $\mathcal{T}=8$, $N=16$, and $K=10$}
%\label{h1}
%\end{figure}
%
%\begin{figure}
%\centering
%\includegraphics[width=3.7in]{images/h2}
%\caption{Comparison of different heuristics with $\mathcal{T}=16$, $N=16$, and $K=10$}
%\label{h2}
%\end{figure}
%
%\begin{figure}
%\centering
%\includegraphics[width=3.7in]{images/h3}
%\caption{Comparison of different heuristics with $\mathcal{T}=24$, $N=16$, and $K=10$}
%\label{h3}
%\end{figure}

\begin{figure*}
\centering{
\subfigure[Comparison of different heuristics with $\mathcal{T}=8$, $N=16$, and $K=10$]{\includegraphics[width=3.5in]{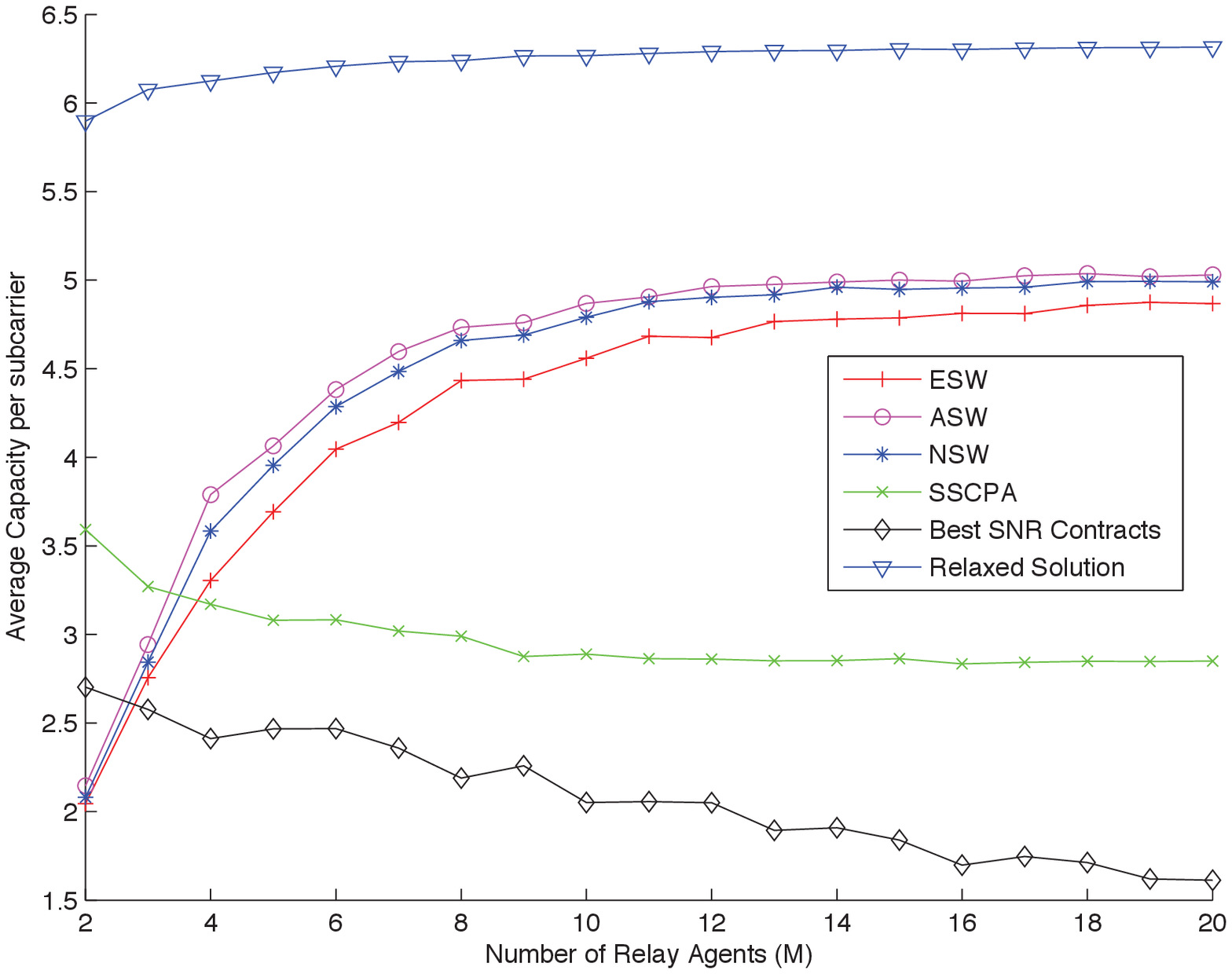}
\label{h1}}
\subfigure[Comparison of different heuristics with $\mathcal{T}=16$, $N=16$, and $K=10$]{\includegraphics[width=3.5in]{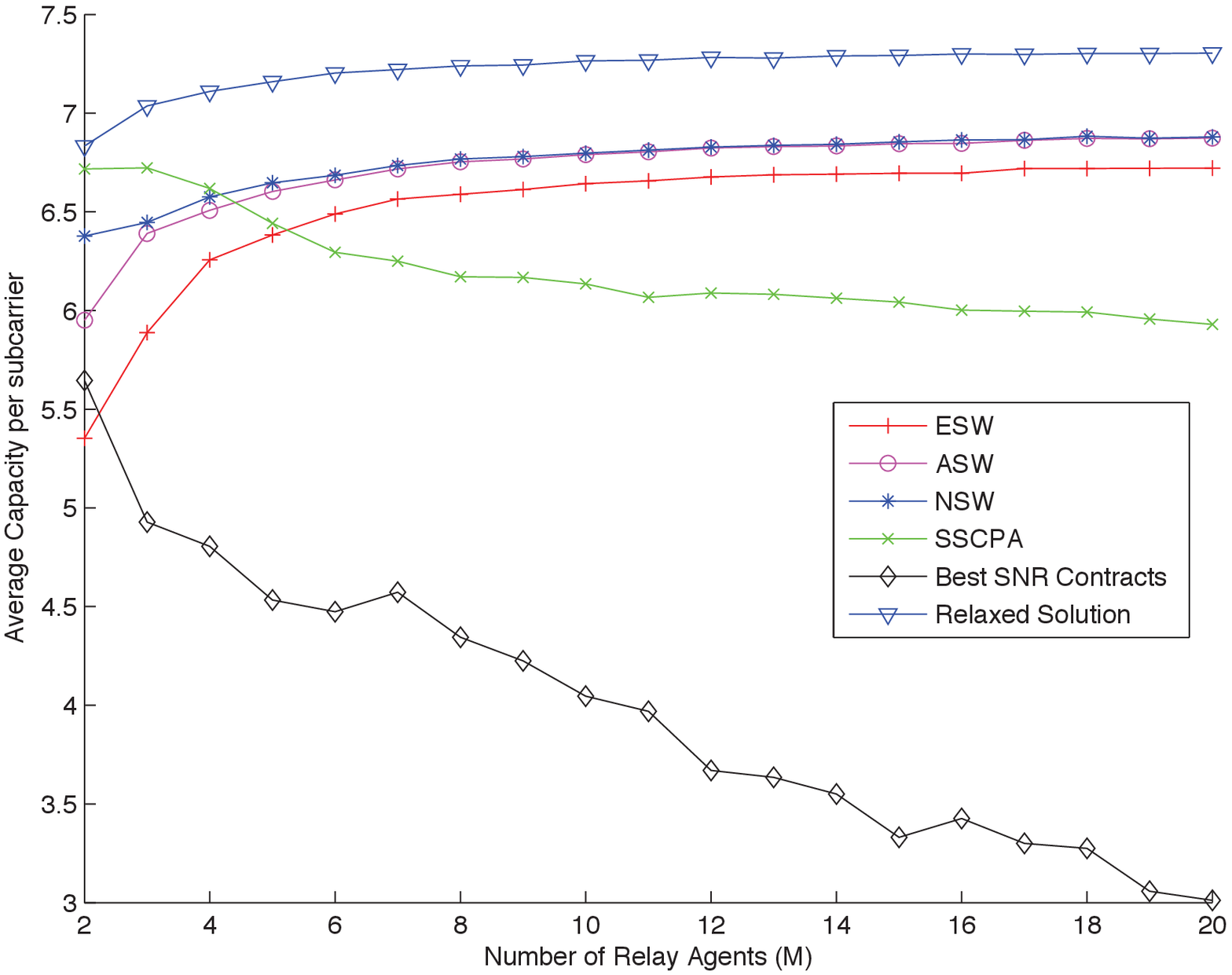}
\label{h2}}
\subfigure[Comparison of different heuristics with $\mathcal{T}=24$, $N=16$, and $K=10$]{\includegraphics[width=3.5in]{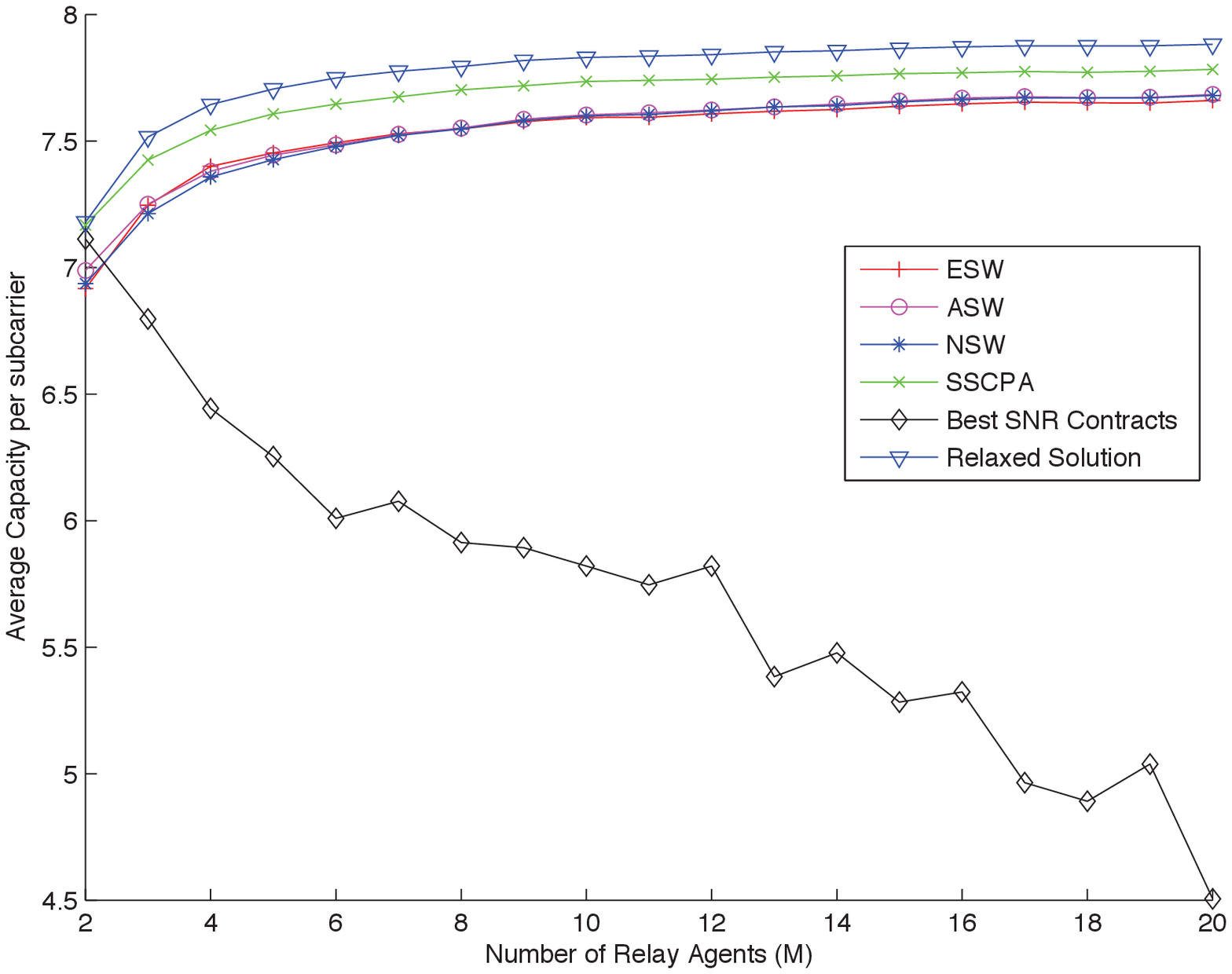}
\label{h3}}
}
\caption{Comparison of heuristic schemes}
\end{figure*}

Fig. \ref{numsub} compares the performance of the proposed heuristic, i.e., overall average capacity with respect to number of subcarriers for two values of budget constraints. The simulation parameters are provided under the figure. We could easily deduce from this figure that the performance of ``Best SNR Contracts" saturates very quickly and is far inferior from the proposed heuristic solution because of lower overall average capacity. One reason why ``Best SNR Contracts" has inferior performance is that in this solution the best SNRs may not be well-spread over all subcarriers and some of the subcarriers may be underused.

%\begin{figure}
%\centering
%\includegraphics[width=3.6in]{images/numsub}
%\caption{Capacity vs. number of subcarriers (N) for fixed $M=10$, and $K=10$ and two values of $\mathcal{T}=16,24$}
%\label{numsub}
%\end{figure}
%
%\begin{figure}
%\centering
%\includegraphics[width=3.6in]{images/quant}
%\caption{Average Capacity vs. Quantization Factor (K) for fixed $N=16$, $M=10$ and two values of $\mathcal{T}=16,24$}
%\label{quant}
%\end{figure}
%
%\begin{figure}
%\centering
%\includegraphics[width=3.6in]{images/fbsb}
%\caption{Comparison of first best and second-best contracts, and full information scenario with $N=16$, $K=10$, and two values of $\mathcal{T}=8,24$}
%\label{fbsb}
%\end{figure}

\begin{figure*}
\centering{
\subfigure[Capacity vs. number of subcarriers (N) for fixed $M=10$, and $K=10$ and two values of $\mathcal{T}=16,24$]{\includegraphics[width=3.5in]{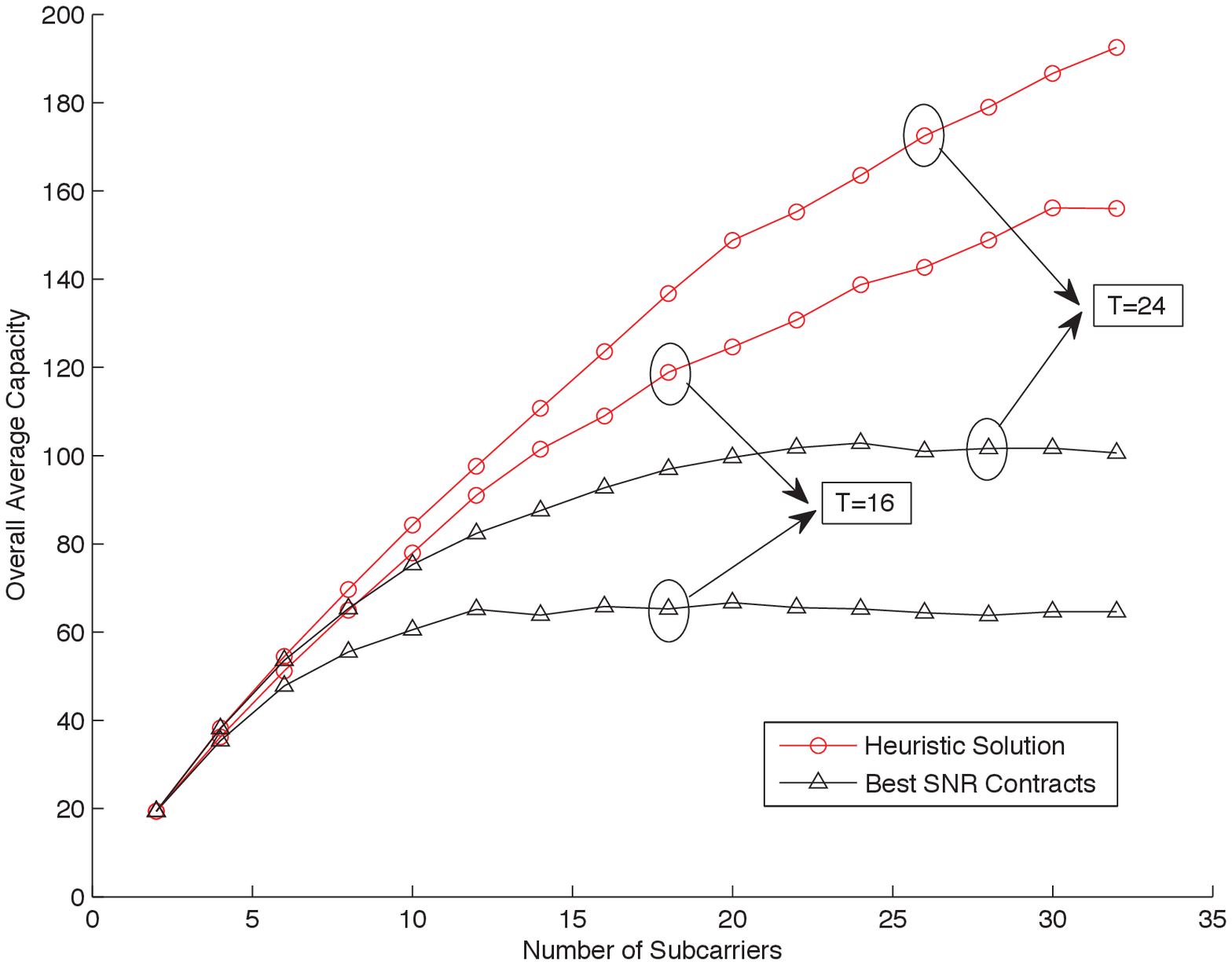}
\label{numsub}}
\subfigure[Average Capacity vs. Quantization Factor (K) for fixed $N=16$, $M=10$ and two values of $\mathcal{T}=16,24$]{\includegraphics[width=3.5in]{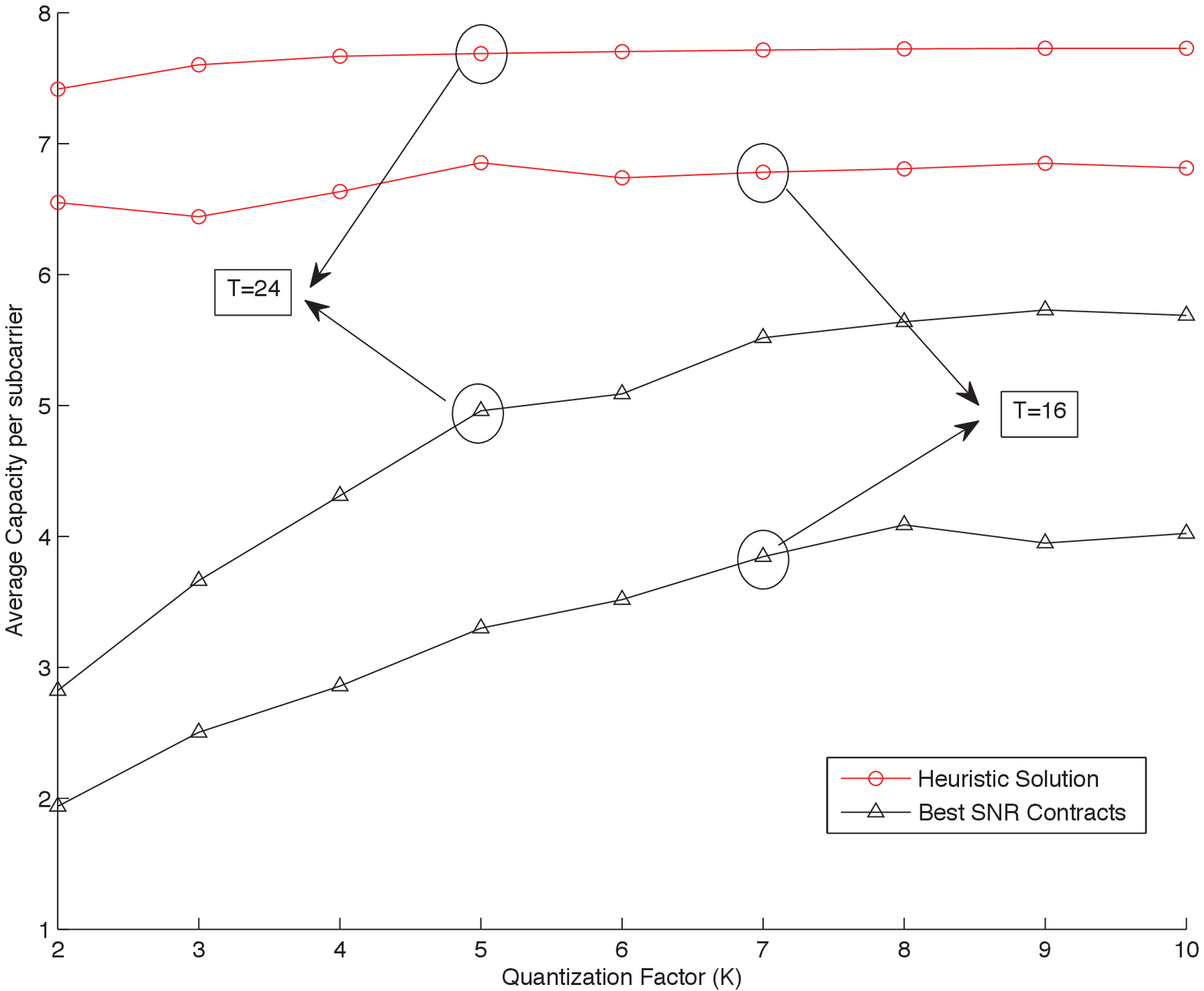}
\label{quant}}
\subfigure[Comparison of first best and second-best contracts, and full information scenario with $N=16$, $K=10$, and two values of $\mathcal{T}=8,24$]{\includegraphics[width=3.5in]{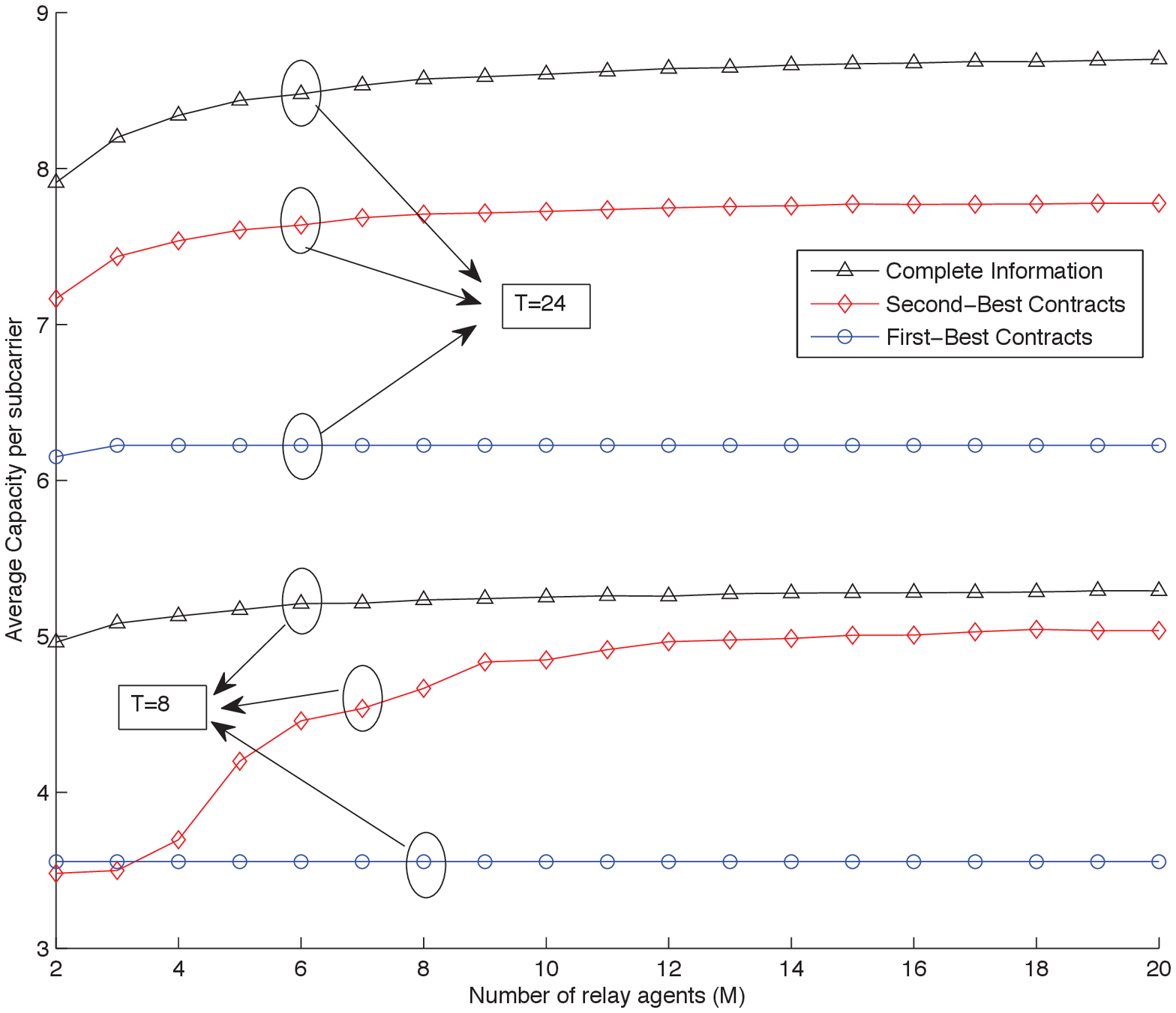}
\label{fbsb}}
}
\caption{Performance evaluation}
\end{figure*}

In Fig. \ref{quant}, we plot the average capacity per subcarrier vs. quantization factor for the heuristic solution and for the ``Best SNR Contracts" solution with two different budget constraints. Once again, the parameters are provided underneath the graph. It is interesting to observe that the average capacity per subcarrier for the proposed heuristic solution remains almost the same as we increase the quantization factor $K$, which means that it is not that advantageous to quantize the probability distribution to a very high factor in order to obtain better performance. For example, a quantization factor as low as 3 which essentially classifies the types as ``good", ``average", or ``bad" can be sufficient. This further demonstrates numerically that the source needs to design and broadcast very few contracts to get fair performance which leads to less signalling overheads. However, this observation is valid for uniform distribution and results could vary for a non-uniform distribution where higher quantization factor may lead to some types being more probable than others.

Lastly, we analyze the performance of our system under both complete and incomplete information scenarios. In order to do that, we plot average capacity per subcarrier vs. the number of relay agents for two values of budget constraints under these two scenarios in Fig. \ref{fbsb}. We use heuristic solution for relay selection to evaluate the capacity. For each budget constraint, we plot the case where we have complete information, i.e., source knows the efficient contracts for all relays, and for incomplete information we plot two cases; when first-best contracts are broadcasted, and when second-best contracts are broadcasted. As expected, on average, for each budget constraint, the performance gap due to information asymmetry between the complete information and the incomplete information with the second-best contracts is smaller than gap between complete information and the incomplete information with the first-best contracts. Moreover, the capacity obtained with the first-best contracts is almost always constant because the first-best contracts are not designed to be incentive compatible and only the smallest contract is selected by all relays, so increasing the number of agents has no effect on performance.

\section{Conclusion}
\label{concl}
In this paper, we study the problem of relay selection and incentive mechanisms in multi-carrier wireless systems under asymmetric information.  In such networks, a source node is assumed to be ill-informed of possible potential relay nodes and their private information such as channel conditions on the relay-destination link. This information asymmetry makes it harder for the source to choose selfish relay nodes efficiently and hence to optimize its throughput. In this paper, we address this classical problem by introducing a simple {\em principal-agent model} for source and relays. The advantage of this model is that it leaves the bargaining power completely to the principal which in our case is the source node, and this reduces signalling and computational overheads for the relays. We then divide the problem into two parts. In the first part, we use contract theory to design a common incentive compatible contracts for the relays (assuming that source has a joint distribution of private information or {\em types} of relays for all subcarriers), i.e., a relay of a certain type will choose contract designed for his type only. These contracts are broadcasted by the source to all relays and the interested relay-agents respond with the contracts they are willing to accept in each subcarrier. Once the source becomes aware of these contracts, the only problem that remains is to select appropriate relays in each subcarrier such that the overall capacity of source is maximized while the source is under a budget constraint. We formulate this problem of relay selection as a nonlinear non-separable knapsack problem and suggest a heuristic solution to solve it efficiently. The source then notifies the selected relays with instructions such as space-time codes and makes the required transfers. We have compared the performance of heuristic solution with a simple relay selection mechanism and have presented numerical results to show that our solution performs better under the most common settings. The overall mechanism introduced in this paper is simple and has limited interaction between the source and the potential relays and participation of interested nodes is also voluntary. As part of future work, we will address more complex issues for the relay selection mechanism such as including the effect of direct links and considering selective relaying in the proposed system model (selective relaying is discussed in \cite{hasanicc}), addressing energy efficiency with asymmetric information (energy efficiency for cooperative networks is explored \cite{green}), and designing better heuristics and more efficient protocols.

\section*{Appendix}
\subsection{Proof of Proposition \ref{prop:contsobjective}}
\label{proof1}
Lets assume the marginal distribution of type $\theta_n$ to be $f_n(\theta_n)$ and the corresponding cumulative distribution to be $F_n(\theta_n)$ (s.t. $F_n'(\theta)=f_n(\theta)$), where $f_n(\theta_n)$ is given by

\begin{dmath}
f_n(\theta_n)=\\\int_{\underline\theta}^{\bar\theta}\int_{\underline\theta}^{\bar\theta}\cdots\int_{\underline\theta}^{\bar\theta}f(\theta_1,\theta_2,\cdots,\theta_N) d\theta_1d\theta_2\cdots\theta_{n-1}\theta_{n+1}\cdots d\theta_N.
\label{marginial}
\end{dmath}

Using the expression (\ref{marginial}), we can simplify (\ref{sobjective}) and subsequently substitute (\ref{tranc}) as described in equation (\ref{eq:simplifiedintegral}).  We can then use integration by parts to rewrite the last term of (\ref{eq:simplifiedintegral}) to obtain (\ref{eq:byparts}). Substituting (\ref{eq:byparts}) in (\ref{eq:simplifiedintegral}), we can reduce the corresponding optimization problem to (\ref{ctobjective}). Q.E.D.

\begin{figure*}
\begin{eqnarray}
&&\int_{\underline\theta}^{\bar\theta}\int_{\underline\theta}^{\bar\theta}\cdots\int_{\underline\theta}^{\bar\theta}\sum_{n=1}^{N}(U(\gamma(\theta_n))-t(\theta_n))f(\theta_1,\theta_2,\cdots,\theta_N)d\theta_1d\theta_2\cdots d\theta_N\nonumber\\
&=&\sum_{n=1}^{N} \int_{\underline\theta}^{\bar\theta}\int_{\underline\theta}^{\bar\theta}\cdots\int_{\underline\theta}^{\bar\theta}(U(\gamma(\theta_n))-t(\theta_n))f(\theta_1,\theta_2,\cdots,\theta_N)d\theta_1d\theta_2\cdots d\theta_N\nonumber\\
&=&\sum_{n=1}^{N}\int_{\underline\theta}^{\bar\theta}(U(\gamma(\theta_n))-t(\theta_n))f_n(\theta_n)d\theta_n
=\sum_{n=1}^{N}\int_{\underline\theta}^{\bar\theta}(U(\gamma(\theta))-t(\theta))f_n(\theta)d\theta\nonumber\\
&=&\sum_{n=1}^{N}\int_{\underline\theta}^{\bar\theta}\left(U(\gamma(\theta))-{c\gamma(\theta)\over\theta}-\int_{\underline\theta}^{\theta}{c\gamma(\tau)\over\tau^2}d\tau\right)f_n(\theta)d\theta
\label{eq:simplifiedintegral}
\end{eqnarray}
\hrulefill
\begin{eqnarray}
\int_{\underline\theta}^{\bar\theta}\int_{\underline\theta}^{\theta}{c\gamma(\tau)\over\tau^2} f_n(\theta)d\tau d\theta
&=&\int_{\underline\theta}^{\bar\theta}F_n'(\theta)\int_{\underline\theta}^{\theta}{c\gamma(\tau)\over\tau^2} d\tau d\theta
=\left. F_n(\theta)\int_{\underline\theta}^{\theta}{c\gamma(\tau)\over\tau^2} d\tau \right|_{\underline\theta}^{\bar\theta}-
\int_{\underline\theta}^{\bar\theta}{c\gamma(\theta)\over\theta^2} F_n(\theta) d\theta\nonumber\\
&=&\int_{\underline\theta}^{\bar\theta}{c\gamma(\theta)\over\theta^2} (1-F_n(\theta)) d\theta
=\int_{\underline\theta}^{\bar\theta}{c\gamma(\theta)\over\theta^2} \frac{1-F_n(\theta)}{f_n(\theta)}f_n(\theta)d\theta
\label{eq:byparts}
\end{eqnarray}

\hrulefill
\end{figure*}

\subsection{Proof of Theorem \ref{th2}}
\label{proof2}
Let us consider two relay-agent types $\delta_k$ and $\delta_j$ such that $\delta_k>\delta_j$. The mutual IC conditions for these relay types are
\begin{eqnarray}
\label{icu}
t_k-{c\gamma_k\over\delta_k}\ge t_j-{c\gamma_j\over\delta_k},\\
t_j-{c\gamma_j\over\delta_j}\ge t_k-{c\gamma_k\over\delta_j}.
\end{eqnarray}
Adding these two IC conditions gives
\begin{eqnarray}
{c\gamma_k\over\delta_k}+{c\gamma_j\over\delta_j}\le{c\gamma_j\over\delta_k}+{c\gamma_k\over\delta_j}\nonumber\\
\text{or, } \delta_j(\gamma_k-\gamma_j)\le\delta_k(\gamma_k-\gamma_j).
\end{eqnarray}
Since $\delta_k>\delta_j$, we can say that $\gamma_k\ge\gamma_j$. And because $\gamma$'s are all assumed to be positive, we have 
 \begin{equation}
 0\le\gamma_1\le\gamma_2\le\cdots\le \gamma_K.
 \label{snrrelation}
 \end{equation}
So, we proved the first result and now we must also prove that for the optimal solution, ICs reduce to (\ref{adjacent}) when compounded with (\ref{snrrelation}). This proof is a slightly more intricate but we will first prove that when the downward adjacent ICs are binding, i.e., (\ref{adjacent}) is true, all the other ICs are automatically satisfied (sufficiency) and we will then prove that when one or more downward adjacent ICs is not binding, the contract cannot be optimal and a better contract can be obtained by binding the downward adjacent ICs recursively by reducing the transfers, hence increasing the source's utility (necessity).

\subsubsection{Proof of Sufficiency} Suppose the adjacent downward ICs are binding, then we can rewrite (\ref{adjacent}) as follows:
\begin{equation}
(t_k-t_{k-1})={c(\gamma_{k}-\gamma_{k-1})\over\delta_k},\;\;\;\forall k\ge2.
\label{icother}
\end{equation}
Lets consider some $j$ such that $j>k$, then using (\ref{snrrelation}) and (\ref{icother}), we can write:
\begin{eqnarray}
&&(t_j-t_k) \nonumber\\           &=&(t_j-t_{j-1})+(t_{j-1}-t_{j-2})+\cdots+(t_{k+1}-t_k)\nonumber\\
			  &=&{c(\gamma_{j}-\gamma_{j-1})\over\delta_j}+{c(\gamma_{j-1}-\gamma_{j-2})\over\delta_{j-1}}+\cdots{c(\gamma_{k+1}-\gamma_{k})\over\delta_{k+1}}\nonumber\\
			  &\le&{c(\gamma_{j}-\gamma_{j-1})\over\delta_k}+{c(\gamma_{j-1}-\gamma_{j-2})\over\delta_{k}}+\cdots{c(\gamma_{k+1}-\gamma_{k})\over\delta_{k}}\nonumber\\
                        &=&{c(\gamma_j-\gamma_k)\over\delta_k}.
                        \label{prhs}                
\end{eqnarray}
Now, lets take some $j$ such that $j<k$, then using (\ref{snrrelation}) and (\ref{icother}), we can write:
\begin{eqnarray}
&&(t_k-t_j)   \nonumber\\            &=&(t_k-t_{k-1})+(t_{k-1}-t_{k-2})+\cdots+(t_{j+1}-t_j)\nonumber\\
			  &=&{c(\gamma_{k}-\gamma_{k-1})\over\delta_k}+{c(\gamma_{k-1}-\gamma_{k-2})\over\delta_{k-1}}+\cdots{c(\gamma_{j+1}-\gamma_{j})\over\delta_{j+1}}\nonumber\\
			  &\ge&{c(\gamma_{k}-\gamma_{k-1})\over\delta_k}+{c(\gamma_{k-1}-\gamma_{k-2})\over\delta_{k}}+\cdots{c(\gamma_{j+1}-\gamma_{j})\over\delta_{k}}\nonumber\\
                        &=&{c(\gamma_k-\gamma_j)\over\delta_k}.
                        \label{plhs}                
\end{eqnarray}
Combining equations (\ref{prhs}) and (\ref{plhs}), we get:
\begin{eqnarray}
(t_k-t_j) \ge{c(\gamma_k-\gamma_j)\over\delta_k} \text{  or  }t_k-{c\gamma_k\over\delta_k}\ge t_{j}-{c\gamma_{j}\over\delta_k},\;\;\;\forall j\ne k.
\end{eqnarray}
Hence all the ICs are automatically satisfied when adjacent downward ICs are binding.

\subsubsection{Proof of Necessity} Suppose one or more downward adjacent ICs are not binding for optimal solution (and IRs are satisfied because it is a solution). Lets take one such type $\delta_k$ for which adjacent downward IC is inactive, so using IR for type $\delta_{k-1}$ we have:
\begin{equation}
t_k-{c\gamma_k\over\delta_k}> t_{k-1}-{c\gamma_{k-1}\over\delta_k}\ge  t_{k-1}-{c\gamma_{k-1}\over\delta_{k-1}}\ge0\;\;\;\forall j\ge k, j\le k
\end{equation}
What this means is that if we reduce all $t_j$'s, $\forall j\ge k$ with equal amount through which the adjacent downward IC for type $\delta_k$ becomes active, it will not affect any of the IRs and the existing relation (whether binding or not binding) between adjacent downward ICs for every other type. We can iteratively repeat the process (starting with the lowest type for which adjacent downward IC is inactive) till all the adjacent downward ICs are binding. In this process, we have only reduced the transfers to bind all the downward adjacent ICs, which in turn automatically guarantees that all the other ICs are satisfied from the sufficiency conditions proved in Part 1. This therefore increases source's profit because transfers are expensive for the source. Hence, the original contract cannot be optimal because we found a better contract, which is a contradiction. 

\subsection{Proof of Proposition \ref{prop:adverse1}}
\label{proof3}
Lets look at the second term $\sum_{n=1}^{N}\sum_{k=1}^{K}\pi_{kn}t_k$ of the objective of optimization problem in (\ref{optimiz}). Substituting the equality constraints from (\ref{optimiz}), we write the individual terms of the summation as described by (\ref{ct:summationseries}). Now adding the individual equations in (\ref{ct:summationseries}), collecting the terms for each $\gamma_k$, and using the fact that $\sum_{i=1}^{K}\pi_{in}=1$ we can derive (\ref{ct:proofprep2}). Substituting (\ref{ct:proofprep2}) back in (\ref{optimiz}), we get (\ref{reducedoptim}). Q.E.D.

\begin{figure*}[!t]
\begin{eqnarray}
\sum_{n=1}^{N}\pi_{1n}t_1&=&\sum_{n=1}^{N}\pi_{1n}\left({c\gamma_1\over\delta_1}\right)\nonumber\\
\sum_{n=1}^{N}\pi_{2n}t_2&=&\sum_{n=1}^{N}\pi_{2n}\left({c\gamma_1\over\delta_1}+{c\gamma_2\over\delta_2}-{c\gamma_1\over\delta_2}\right)\nonumber\\
\sum_{n=1}^{N}\pi_{3n}t_3&=&\sum_{n=1}^{N}\pi_{3n}\left({c\gamma_1\over\delta_1}+{c\gamma_2\over\delta_2}-{c\gamma_1\over\delta_2}+{c\gamma_3\over\delta_3}-{c\gamma_2\over\delta_3}\nonumber\right)\\
\vdots&&\hspace{1 in}\vdots\nonumber\\
\sum_{n=1}^{N}\pi_{Kn}t_K&=&\sum_{n=1}^{N}\pi_{Kn}\left({c\gamma_1\over\delta_1}+{c\gamma_2\over\delta_2}-{c\gamma_1\over\delta_2}+{c\gamma_3\over\delta_3}-{c\gamma_2\over\delta_3}\cdots+{c\gamma_K\over\delta_K}-{c\gamma_{K-1}\over\delta_K}\right).
\label{ct:summationseries}
\end{eqnarray}
\hrulefill
\begin{eqnarray}
\sum_{n=1}^{N}\sum_{k=1}^{K}\pi_{kn}t_k&=&c\gamma_1 \left({1\over\delta_1}\sum_{n=1}^{N}\sum_{i=1}^{K}\pi_{in}-{1\over\delta_2}\sum_{n=1}^{N}\sum_{i=2}^{K}\pi_{in}\right)+
                                    c\gamma_2 \left({1\over\delta_2}\sum_{n=1}^{N}\sum_{i=2}^{K}\pi_{in}-{1\over\delta_3}\sum_{n=1}^{N}\sum_{i=3}^{K}\pi_{in}\right)\nonumber\\
                                    &&+\cdots c\gamma_k\left({1\over\delta_k}\sum_{n=1}^{N}\sum_{i=k}^{K}\pi_{in}-{1\over\delta_{k+1}}\sum_{n=1}^{N}\sum_{i=k+1}^{K}\pi_{in}\right)+\cdots
                                    +\sum_{n=1}^{N}\pi_{Kn}{c\gamma_K \over\delta_K}\nonumber\\
                                    &=&\sum_{k=1}^{K-1}c\gamma_k \left({1\over\delta_k}\sum_{n=1}^{N}\sum_{i=k}^{K}\pi_{in}-{1\over\delta_{k+1}}\sum_{n=1}^{N}\sum_{i=k+1}^{K}\pi_{in}\right)+\sum_{n=1}^{N}\pi_{Kn}{c\gamma_K \over\delta_K}\nonumber\\
                                     &=&\sum_{n=1}^{N}\sum_{k=1}^{K-1}\pi_{kn}c\gamma_k \left({1\over\delta_k}+\left({1-\sum_{i=1}^{k}\pi_{in}\over\pi_{kn}}\right)\left({1\over\delta_k}-{1\over\delta_{k+1}}\right)\right)+\sum_{n=1}^{N}\pi_{Kn}{c\gamma_K \over\delta_K}
\label{ct:proofprep2}
\end{eqnarray}
\hrulefill
\end{figure*}

\begin{biography}[{\includegraphics[width=1in,height=1.25in,clip,keepaspectratio]{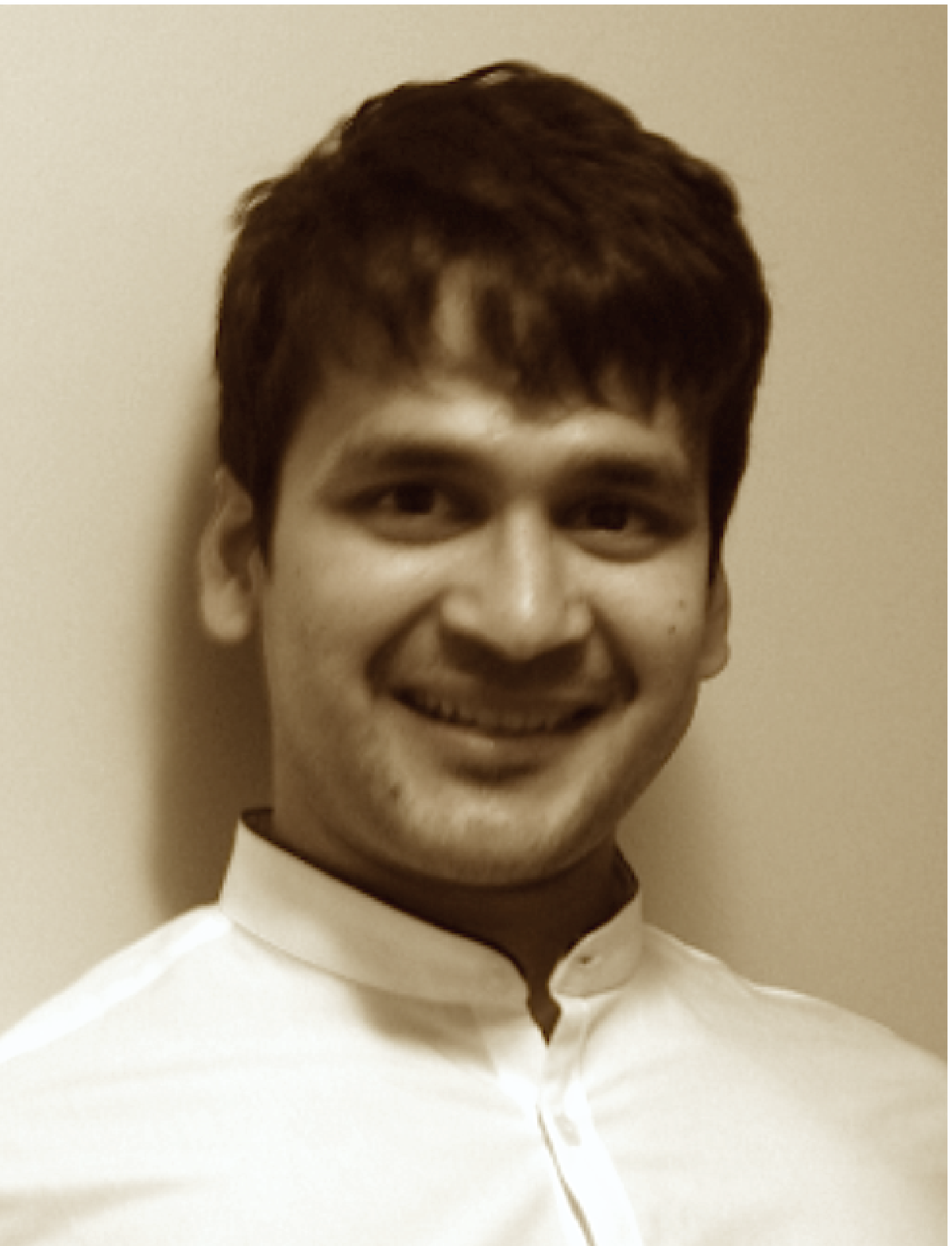}}] {Ziaul Hasan}
received his B.Tech. degree in Electrical Engineering from the Indian Institute of Technology (IIT), Kanpur, India, in 2005 and M.A.Sc and Ph.D. degrees in Electrical and Computer Engineering from University of British Columbia (UBC), Canada in 2008 and 2013 respectively. During his Ph.D. studies, he was awarded Alexander Graham Bell Canada Graduate Scholarship for Doctoral Studies (CGS-D) by Natural Sciences and Engineering Research Council of Canada (NSERC), along with UBC Faculty of Applied Science Graduate Award. In 2011, he also received Michael Smith Foreign Study Supplement Award by NSERC to conduct research at University of Sydney, Australia. Prior to UBC, he has worked as a Senior Associate for Headstrong Corporation, India between 2005 to 2006. In 2009, he worked as a research intern for Broadcom Corporation, Canada. His current research interests are in the areas of resource management for cooperative communications and cognitive radios, heterogenous networks, and economic modelling of wireless networks.
\end{biography}

\begin{biography}[{\includegraphics[width=1in,height=1.25in,clip,keepaspectratio]{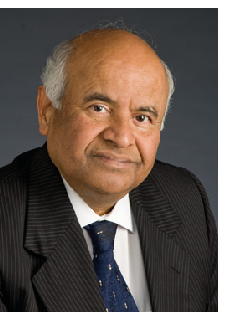}}] {Vijay K. Bhargava}
is a Professor in the Department of Electrical and Computer Engineering at the University of British Columbia in Vancouver, where he served as Department Head during 2003-2008. An active researcher, Vijay is currently leading a major R\&D program in Cognitive and Cooperative Wireless Communication Networks. He received his PhD from Queen's University in 1974. Vijay has held visiting appointments at Ecole Polytechnique de MontrŽal, NTT Research Lab, Tokyo Institute of Technology, Tohoku University and the Hong Kong University of Science and Technology. He is in the Institute for Scientific Information (ISI) Highly Cited list. 

Vijay is a co-author (with D. Haccoun, R. Matyas and P. Nuspl) of ``{\em Digital Communications by Satellite}" (New York: Wiley: 1981) which was translated in Chinese and Japanese. He is a co-editor (with S. Wicker) of ``{\em Reed Solomon Codes and their Applications}" (IEEE Press: 1994), a co-editor (with V. Poor, V. Tarokh and S. Yoon) of ``{\em Communications, Information and Network Security}" (Kluwer: 2003) a co-editor (with E. Hossain) of "Cognitive Wireless Communication Networks" (Springer: 2007), a co-editor (with E. Hossain and D.I Kim) of ``{\em Cooperative Wireless Communications Networks}", (Cambridge University Press: 2011) and a co-editor (with E. Hossain and G. Fettweis) of ``{\em Green Radio Communications Networks}, (Cambridge University Press: 2012)

Vijay is a Fellow of the IEEE, The Royal Society of Canada, The Canadian Academy of Engineering and the Engineering Institute of Canada. He is a Foreign Fellow of the National Academy of Engineering (India) and has served as a Distinguished Visiting Fellow of the Royal Academy of Engineering (U.K.). He has received numerous awards for his teaching, research and service.

Vijay has served on the Board of Governors of the IEEE Information Theory Society and the IEEE Communications Society. He has served as an Editor of the {\em IEEE Transactions on Communications}. He played a major role in the creation of the IEEE Communications and Networking Conference (WCNC) and {\em IEEE Transactions on Wireless Communications}, for which he served as the editor-in-chief during 2007, 2008 and 2009. He is a past President of the IEEE Information Theory Society and is currently serving as the President of the IEEE Communications Society. 
\end{biography}

\begin{thebibliography}{1}

\bibitem{3gpp}
S. W. Peters {\em et. al}, ``Relay architectures for 3GPP LTE-Advanced," {\em EURASIP Journal on Wireless Communications and Networking}, vol. 2009, Article ID 618787, 14 pages, 2009.

\bibitem{802}
IEEE 802.16's Relay Task Group: http://www.ieee802.org/16/relay/

\bibitem{zhuang}
W. Zhuang, and M. Ismail, ``Cooperation in wireless communication networks," {\em IEEE Wireless Commun.}, vol. 19, no. 2, pp. 10-20, 2012.

\bibitem{sendonaris}
A. Sendonaris, E. Erkip, and B. Aazhang, ``User cooperation diversity, part I: System description," {\em IEEE Trans. Commun.}, vol. 51, no. 11, pp. 1927-1938, Nov. 2003.

\bibitem{sadek}
A. Ibrahim, A. Sadek, Weifeng Su, and K. Liu, ``Cooperative communications with relay-selection: when to cooperate and whom to cooperate with?," {\em IEEE Trans. on Wireless Commun.}, vol. 7, no. 7, pp. 2814-2827, 2008.

\bibitem{krikidis}
I~Krikidis, T Charalambous, and J. S Thompson, ``Buffer-Aided Relay Selection for Cooperative Diversity Systems without Delay Constraints," {\em IEEE Trans. on Wireless Commun.}, vol.11, no.5, pp.1957-1967, May 2012.

\bibitem{beres}
E.~Beres, and R.~Adve, ``Selection cooperation in multi-source cooperative networks," {\em IEEE Trans. on Wireless Commun.}, vol. 7, pp. 118-127, Jan. 2008.

\bibitem{ahmed}
I.~Ahmed {\em et. al}, ``Relay Subset Selection and Fair Power Allocation for Best and Partial Relay Selection in Generic Noise and Interference," {\em IEEE Trans. on Wireless Commun.}, vol.11, no.5, pp.1828-1839, May 2012.

\bibitem{prs}
I. Krikidis {\em et. al}, ``Amplify-and-forward with partial relay selection," {\em IEEE Commun. Letters}, vol.12, no.4, pp.235-237, April 2008.

\bibitem{dejun}
Dejun Yang, Xi Fang, and Guoliang Xue, ``Game theory in cooperative communications," {\em IEEE Wireless Commun.}, vol.19, no.2, pp.44-49, April 2012.

\bibitem{riihonen}
T. Riihonen, R. Wichman, and S. Werner, ``Evaluation of OFDM(A) Relaying Protocols: Capacity Analysis in Infrastructure Framework," {\em IEEE Trans. on Vehicular Tech.}, vol.61, no.1, pp.360-374, Jan. 2012.

\bibitem{herbert}
Herbert Gintis, {\em Game Theory Evolving: A Problem-Centered Introduction to Modeling Strategic Interaction}, Princeton University Press, 2009.

\bibitem{salanie}
Bernard Salani\'e, {\em The Economics of Contracts: A Primer}, Cambridge, MA: MIT Press, March 2005.

\bibitem{duan}
L. Duan, L. Gao, and J. Huang, ``Contract-based cooperative spectrum sharing," in {\em Proc. of IEEE DySPAN}, pp.399-407, May 2011.

\bibitem{kalathil}
D. M. Kalathil, and R. Jain, ``Spectrum Sharing through Contracts," in {\em Proc. of IEEE DySPAN}, pp. 1-9, 2010.

\bibitem{kalathil2}
D. M. Kalathil, and R. Jain, ``A contracts-based approach for spectrum sharing among cognitive radios," in {\em Proc. of WiOpt}, pp. 91-97, 2010.

\bibitem{gao}
L. Gao, X. Wang, Y. Xu, and Q. Zhang, ``Spectrum Trading in Cognitive Radio Networks: A Contract-Theoretic Modeling Approach," {\em IEEE Journal on Selected Areas in Commun.}, vol. 29, no. 4, pp. 843-855, 2011.

\bibitem{yangyang}
Y. Yang, H. Hu, J. Xu, and G. Mao, ``Relay technologies for WiMax and LTE-advanced mobile systems," {\em IEEE Commun. Magazine}, vol.47, no.10, pp.100-105, October 2009.

\bibitem{laneman}
J. N. Laneman, and G. W. Wornell, ``Distributed space-time-coded protocols for exploiting cooperative diversity in wireless networks," {\em IEEE Trans. on Information Theory}, vol.49, no.10, pp. 2415-2425, Oct. 2003.

\bibitem{molisch}
A. F. Molisch, {\em Wireless Communications}, 2nd ed., Hoboken, NJ: Wiley, 2011.

\bibitem{brett}
K.~M.~Bretthauer, and B.~Shetty, ``The nonlinear knapsack problem-algorithms and applications," {\em European Journal of Operational Research}, vol. 138, no. 3, pp. 459-472, 2002.

\bibitem{hillier}
F. S. Hillier, {\em Introduction to Operations Research}, New York: McGraw Hill, 2004.

\bibitem{sharkey}
T. C. Sharkey, H. E. Romeijn, and J. Geunes, ``A class of nonlinear nonseparable continuous knapsack and multiple-choice knapsack problems," {\em Math. Program.}, vol. 126, no. 1, pp. 69-96, Mar. 2009.

\bibitem{romeijn}
H. E. Romeijn, J. Geunes, and K. Taaffe, ``On a nonseparable convex maximization problem with continuous knapsack constraints," {\em Operations Research Letters}, vol. 35, no. 2, pp. 172-180, Mar. 2007.

\bibitem{korte}
B. H. Korte, and J. Vygen, {\em Combinatorial optimization: Theory and Algorithms}, Berlin: Springer-Verlag, 2008.

\bibitem{bolton}
Patrick Bolton, {\em Contract Theory}, Cambridge, MA: MIT Press, 2004.

\bibitem{hasanicc}
Z. Hasan, E. Hossain, and V. K. Bhargava, ``Resource Allocation for Multiuser OFDMA-Based Amplify-and-Forward Relay Networks with Selective Relaying," in {\em Proc. of ICC}, pp.1-6, 5-9 June 2011.

\bibitem{green}
Z. Hasan, H. Boostanimehr, and V. K. Bhargava, ``Green Cellular Networks: A Survey, Some Research Issues and Challenges," {\em IEEE Commun. Surveys \& Tutorials}, vol.13, no.4, pp. 524-540, Fourth Quarter 2011.

\end{thebibliography}
\end{document}